\documentclass[10pt]{article}

\usepackage[colorlinks,citecolor=blue]{hyperref}
\usepackage{amsfonts,amsmath,amssymb,amsthm}
\usepackage[pdftex]{graphicx}
\usepackage[hmargin=1in,vmargin=1in]{geometry}
\usepackage{natbib,setspace,enumerate,mathtools}
\allowdisplaybreaks
\usepackage[toc,title,titletoc,header]{appendix}
\usepackage{etoolbox} 
\usepackage{booktabs}
\usepackage{threeparttable}
\usepackage{adjustbox}
\usepackage{multirow}
\def\qed{\rule{2mm}{2mm}}
\def\indep{\perp \!\!\! \perp}

\newtheorem{theorem}{Theorem}[section]
\newtheorem{lemma}{Lemma}[section]

\newtheorem{corollary}{Corollary}[section]

\theoremstyle{definition}

\newtheorem{remark}{Remark}[section]
\newtheorem{assumption}{Assumption}[section]

\AtEndEnvironment{remark}{~\qed}
\AtEndEnvironment{example}{~\qed}

\DeclareMathOperator*{\var}{Var}
\DeclareMathOperator*{\cov}{Cov}

\newcommand{\cp}{\stackrel{P}{\rightarrow}}
\newcommand{\cd}{\stackrel{d}{\rightarrow}}

\begin{document}

\author{
Yuehao Bai \\
Department of Economics\\
University of Southern California \\
\url{yuehao.bai@usc.edu}
\and
Hongchang Guo \\
Department of Economics\\
Northwestern University\\
\url{hongchangguo2028@u.northwestern.edu}
\and
Azeem M.\ Shaikh\\
Department of Economics\\
University of Chicago \\
\url{amshaikh@uchicago.edu}
\and
Max Tabord-Meehan\\
Department of Economics\\
University of Chicago \\
\url{maxtm@uchicago.edu}
}

\bigskip

\title{Inference in Experiments with Matched Pairs \\ and Imperfect Compliance\thanks{We thank Alex Torgovitsky for helpful discussions. The fourth author acknowledges support from NSF grant SES-2149408.}}

\maketitle

\vspace{-0.3in}

\begin{spacing}{1.2}
\begin{abstract}
This paper studies inference for the local average treatment effect in randomized controlled trials with imperfect compliance where treatment status is determined according to ``matched pairs.''   By ``matched pairs,'' we mean that units are sampled i.i.d.\ from the population of interest, paired according to observed, baseline covariates and finally, within each pair, one unit is selected at random for treatment.  Under weak assumptions governing the quality of the pairings, we first derive the limit distribution of the usual Wald (i.e., two-stage least squares) estimator of the local average treatment effect.  We show further that conventional heteroskedasticity-robust estimators of the Wald estimator's limiting variance are generally conservative, in that their probability limits are (typically strictly) larger than the limiting variance.  We therefore provide an alternative estimator of the limiting variance that is consistent.  Finally, we consider the use of additional observed, baseline covariates not used in pairing units to increase the precision with which we can estimate the local average treatment effect.  To this end, we derive the limiting behavior of a two-stage least squares estimator of the local average treatment effect which includes both the additional covariates in addition to pair fixed effects, and show that its limiting variance is always less than or equal to that of the Wald estimator.  To complete our analysis, we provide a consistent estimator of this limiting variance.  A simulation study confirms the practical relevance of our theoretical results.  Finally, we apply our results to revisit a prominent experiment studying the effect of macroinsurance on microenterprise in Egypt.

\end{abstract}
\end{spacing}



\thispagestyle{empty} 
\newpage
\setcounter{page}{1}

\section{Introduction}

This paper studies inference for the local average treatment effect in randomized controlled trials with imperfect compliance, when treatment status is determined according to a ``matched pairs" design. By ``matched pairs,'' we mean that units are sampled i.i.d.\ from the population of interest, paired according to observed, baseline covariates and finally, within each pair, one unit is selected at random for treatment. This method is used routinely in all parts of the sciences. Indeed, commands to facilitate its implementation are included in popular software packages, such as \texttt{sampsi} in Stata. References to a variety of specific examples can be found, for instance, in the following textbook treatments of randomized experiments: \cite{glennerster2014running}, \cite{riach2002field}, \cite{rosenberger2015randomization}. See also \cite{bruhn2009pursuit}, who, based on a survey of selected development economists, report that 56\% of researchers have used such a design at some point. In many such experiments, compliance may be imperfect:  some recent examples of experiments featuring both matched pairs and imperfect compliance are  \cite{groh2016macroinsurance} and \cite{resnjanskij2021can}.  Under weak assumptions that ensure pairs are formed so that units within pairs are suitably ``close'' in terms of observed, baseline covariates, we derive a variety of results pertaining to inference about the local average treatment effect in such experiments.  

We first study the behavior of the usual Wald (i.e., two-stage least squares) estimator of the local average treatment effect.  When all observed, baseline covariates are used in forming pairs, we find that the estimator is efficient among all estimators for the local average treatment effect in the sense that it achieves the lower bound on the limiting variance developed in \cite{bai2023efficient} over a broad class of treatment assignment schemes that hold the marginal probability of treatment assignment equal to one half, thereby including matched pairs, in particular, as a special case.  On the other hand, we find that the conventional heteroskedasticity-robust estimator obtained from a two-stage least squares regression (with and without pair fixed effects) is conservative in that its limit in probability is always weakly larger than the limiting variance, and strictly larger unless treatment effect heterogeneity is constrained in a particular fashion. As a result, we provide an alternative estimator of the limiting variance and show that it is consistent. In our simulation study, we find that tests using the Wald estimator together with conventional heteroskedasticity-robust estimators may as a consequence have worse power when compared to tests using the Wald estimator together with our estimator of its limiting variance. 

Next, we analyze the behavior of covariate adjusted Wald estimators for settings in which there are additional observed baseline covariates that were not used when pairing units.  We first derive the limiting behavior of a class of covariate-adjusted estimators indexed by different ``working models'' for the conditional expectations of the outcome and treatment take-up with respect to the covariates.  Importantly, these working models need not be correctly specified for the true conditional expectations in order for this estimator to remain consistent for the local average treatment effect, but we show that the limiting variance of the estimator is minimized, in particular, when they are correctly specified. We then specialize these results to the case where the working models are linear as a function of the covariates, and derive the form of the optimal linear working model. We further show that the resulting optimal linear covariate adjusted estimator can be computed using a two-stage least squares estimator using treatment assignment as an instrument for treatment status in an instrumental variables regression of the outcome on the following quantities: treatment status; (functions of) the observed baseline covariates; and pair fixed effects. We emphasize, however, that the resulting heteroskedasticity robust variance estimator is not guaranteed to be consistent, and so we also provide a suitable consistent estimator of the asymptotic variance. In our simulation study, we find that tests using our optimal linear covariate adjusted estimator together with our consistent estimator of its limiting variance have better power than tests constructed using the unadjusted Wald estimator or tests constructed using sub-optimal linear covariate adjustments.

Our paper builds upon the analysis of \cite{bai2022mp}, who analyzed the behavior of the difference-in-means estimator of the average treatment effect in context of experiments with matched pairs and perfect compliance.  Our covariate-adjusted estimator is inspired by the analysis in \cite{bai2023covariate}, who studied the use of additional, observed, baseline covariates in experiments with matched pairs and perfect compliance to improve the precision with which we can estimate the average treatment effect \citep[similar results have also been obtained by][]{cytrynbaum2023covariate}.  We emphasize, however, that none of the aforementioned papers permit imperfect compliance, which, as argued by \cite{athey2017econometrics}, is one of the most common complications in even the most well designed experiments.  We note, however, that imperfect compliance has been studied in the context of randomized controlled trials with other treatment assignment schemes, such as stratified block randomization: see, for example, \cite{ansel2018ols}, \cite{bugni2021inference}, and \cite{jiang2022improving}. We also emphasize that all of these papers, like ours, carry out their analysis in a superpopulation sampling framework.  In this way, our analysis differs from the analysis of experiments with a finite population sampling framework \citep[see for instance][among many others]{athey2017econometrics, chaisemartin2022at, ding2017bridging, ren2021model}.


The remainder of the paper is organized as follows.  In Section \ref{sec:setup} we describe our setup and notation.  In particular, there we describe the precise sense in which we require that units in each pair are ``close'' in terms of their baseline covariates. Our main results concerning the Wald estimator are contained in Section \ref{sec:main}.  In Section \ref{sec:adjustment}, we develop results pertaining to our covariate-adjusted estimator that exploits additional observed, baseline covariates not used in pairing units.  In Section \ref{sec:sims}, we examine the practical relevance of our theoretical results via a small simulation study. In Section \ref{sec:empirical-application}, we provide a brief empirical illustration of our proposed tests using data from an experiment in \cite{groh2016macroinsurance}.  Finally, we conclude in Section \ref{sec:recs} with some recommendations for empirical practice guided by both our theoretical results and our simulation study.  As explained further in that section, we do not recommend the use of the Wald estimator with the conventional heteroskedasticity-robust estimator of its limiting variance because it is conservative in the sense described above; we instead encourage the use of the Wald estimator with our consistent estimator of its limiting variance because it is asymptotically exact, and, as a result, can be considerably more powerful.  When there are additional, observed, baseline covariates that are not used when forming pairs, we recommend the use of our covariate-adjusted Wald estimator with our consistent estimator of its limiting variance. Proofs of all results are provided in the  Appendix.

\section{Setup and Notation} \label{sec:setup}

Let $Y_{i} \in \mathbf{R}$ denote the (observed) outcome of the $i$th unit, $A_i \in \{0, 1\}$ be an indicator for whether or not unit $i$ is assigned to treatment, $D_i \in \{0, 1 \}$ be an indicator for whether or not unit $i$ decides to take up treatment, $X_i \in \mathbf{R}^{k_x}$ denote observed, baseline covariates for the $i$th unit which are used for matching, and $W_i \in \mathbf{R}^{k_w}$ denote observed, baseline covariates for the $i$th unit which will be used when we consider covariate adjustment in Section \ref{sec:adjustment}. In contrast to the setting considered in \cite{bai2022mp}, we allow for imperfect compliance, i.e. for $D_i \ne A_i$. Further denote by $Y_i(d)$ the potential outcome of the $i$th unit if they make treatment decision $d \in \{0, 1\}$, and by $D_i(a)$ the potential treatment decision of the $i$th unit if assigned to treatment $a \in \{0, 1\}$. The observed treatment decision and potential treatment decision are related to treatment assignment via the usual relationship 
\begin{equation}\label{eq:obsD}
D_i = D_i(1)A_i + D_i(0)(1 - A_i)~,
\end{equation}
and the observed outcome and potential outcome are related to treatment decision via the relationship
\begin{equation}\label{eq:obsY}
Y_i = Y_i(1)D_i + Y_i(0)(1 - D_i)~.
\end{equation}

\noindent We will also often make use of the following alternative representation for the observed outcome, which is numerically equivalent to (\ref{eq:obsY}):
\begin{equation}\label{eq:obsY'}
Y_i = \tilde Y_i(1)A_i + \tilde Y_i(0)(1 - A_i)~,
\end{equation}
where
\begin{equation}\label{eq:TildeY}
\tilde Y_i(a) = Y_i(1)D_i(a) + Y_i(0)(1 - D_i(a))
\end{equation}
for $a \in \{0, 1\}$. In words, $\tilde Y_i(a)$ represents the ``intention-to-treat" potential outcome for unit $i$ when assigned to treatment $a \in \{0, 1\}$. 

Following \cite{angristimbens1994late}, each participant in the experiment can be categorized into one of four types: units for which $D_i(1) = 1$ and $D_i(0) = 0$ are referred to as compliers, units for which $D_i(1) = 1$ and $D_i(0) = 1$ are referred to as always takers, units for which $D_i(1) = 0$ and $D_i(0) = 0$ are referred to as never takers, and finally units for which $D_i(1) = 0$ and $D_i(0) = 1$ are referred to as defiers.  We use the notation
\begin{equation}\label{eq:compliers}
C_i = I\left\{ D_i(1) = 1, D_i(0) = 0 \right\}
\end{equation}
below to indicate whether or not unit $i$ is a complier.

Throughout the paper we will study inference on samples with $2n$ observations, so that $n$ indexes the number of pairs of observations. For a random variable indexed by $i$, say for example $A_i$, it will be useful to denote by $A^{(n)}$ the random vector $(A_1, ..., A_{2n})$. Denote by $P_n$ the distribution of the observed data $(Y^{(n)}, D^{(n)}, A^{(n)}, X^{(n)}, W^{(n)})$, and by $Q_n$ the distribution of $(Y^{(n)}(1), Y^{(n)}(0), D^{(n)}(1), D^{(n)}(0), X^{(n)}, W^{(n)})$. Note that $P_n$ is jointly determined by (\ref{eq:obsD}), (\ref{eq:obsY}), $Q_n$, and the mechanism for determining treatment assignment. We assume that our sample consists of $2n$ i.i.d.\ observations i.e. that $Q_n = Q^{2n}$, where $Q$ is the marginal distribution of $(Y_i(1), Y_i(0), D_i(1), D_i(0), X_i, W_i)$. We therefore state our assumptions below in terms of assumptions on $Q$ and the mechanism for determining treatment assignment. Indeed, we will not make reference to $P_n$ in the sequel and all operations are understood to be under $Q$ and the mechanism for determining treatment assignment.

Our object of interest is the local average treatment effect, which may be expressed in our notation as
\begin{equation} \label{eq:LATE}
\Delta(Q) = E\left[ Y_i(1) - Y_i(0) | C_i = 1 \right].
\end{equation}

For a pre-specified choice of $\Delta_0$, the testing problem of interest is 
\begin{equation}\label{eq:H0}
H_0: \Delta(Q) = \Delta_0\  \text{versus}\  H_1: \Delta(Q) \neq \Delta_0
\end{equation}
at level $\alpha \in (0, 1)$.

We begin by describing our primary assumptions on the data generating process.

\begin{assumption}\label{ass:Q}
The distribution $Q$ is such that \vspace{-.25cm}
\begin{enumerate}[\rm (a)]
\item $0 < E\left[ \var\left[ \tilde{Y}_i(a) - \Delta(Q)D_i(a) | X_i \right] \right]$ for $a \in \{0, 1  \}$.
\item $E\left[ Y_i(a)^2 \right] < \infty$ for $a \in \{0, 1 \}$. 
\item $E[ Y_i(1)^r D_i(a) | X_i = x ]$ and $E[ Y_i(0)^r (1 - D_i(a)) | X_i = x]$ are Lipschitz for $a = 0, 1$ and $r = 0, 1, 2$.
\item $P\{D_i(1) \geq D_i(0)\} = 1$.
        \item $P \{C_i = 1\} > 0$.
    \end{enumerate}
    \end{assumption}

Assumption \ref{ass:Q}(a)--(b) are mild restrictions imposed to rule out degenerate situations and to permit the application of suitable laws of large numbers and central limit theorems. Assumption \ref{ass:Q}(c) is a smoothness requirement that ensures that units that are ``close" in terms of their baseline covariates are suitably comparable. Similar smoothness requirements are also considered in \cite{bai2022mp}, and generally play a key role in establishing the asymptotic exactness of our proposed tests. Assumptions \ref{ass:Q}(d)--(e) are the standard ``monotonicity" and ``relevance" conditions of \cite{angristimbens1994late} which ensure that the probability limit of the Wald estimator which we define in Section \ref{sec:main} can be interpreted as the local average treatment effect $\Delta(Q)$.

Next, we describe our assumptions on the mechanism determining treatment assignment. Following the notation in \cite{bai2022mp}, the $n$ pairs can be represented by the sets
\begin{align*}
\{ \pi(2j - 1), \pi(2j)\} \ \text{for}\ j = 1, ..., n~,
\end{align*}
where $\pi = \pi_n\left(X^{(n)}\right)$ is a permutation of $2n$ elements. Given such a $\pi$, we assume that treatment status is assigned as described in the following assumption:
\begin{assumption}\label{ass:assignment}
Treatment status is assigned so that
\[ \left( Y^{(n)}(1), Y^{(n)}(0), D^{(n)}(1), D^{(n)}(0), W^{(n)} \right) \indep A^{(n)}|X^{(n)}~, \]
and conditional on $X^{(n)}$, $(A_{\pi(2j-1)}, A_{\pi(2j)})$, $j = 1, ..., n$ are i.i.d.\ and each uniformly distributed over the values in $\{(0, 1), (1, 0)\}$.
\end{assumption}

We further require that the units in each pair be ``close" in terms of their baseline covariates in the following sense:

\begin{assumption}\label{ass:pairs}
The pairs used in determining treatment status satisfy 
\begin{align*}
\frac{1}{n} \sum_{1\leq j\leq n} \| X_{\pi(2j)} - X_{\pi(2j-1)} \|_2^r \cp 0~,
\end{align*}
for $r\in\{1, 2\}$.
\end{assumption}

We will also sometimes require that the distances between units in adjacent pairs be ``close" in terms of their baseline covariates:
\begin{assumption}\label{ass:pairsofpairs}
The pairs used in determining treatment status satisfy
\begin{align*}
\frac{1}{n}\sum_{1\leq j\leq \lfloor \frac{n}{2} \rfloor} \| X_{\pi(4j-k)} - X_{\pi(4j-l)} \|_2^2 \cp 0~,
\end{align*}
for $k\in\{ 2,3\}$ and $l\in \{ 0,1\}$.
\end{assumption}

\cite{bai2022mp}, \cite{bai2023inference}, and \cite{cytrynbaum2023designing} provide several examples of pairing algorithms which satisfy Assumptions \ref{ass:pairs}--\ref{ass:pairsofpairs}. The simplest such example is when $X_i \in \mathbf{R}$, in which case we can order units from smallest to largest according to $X_i$ and pair adjacent units. It then follows from Theorem 4.1 in \cite{bai2022mp} that Assumptions \ref{ass:pairs}--\ref{ass:pairsofpairs} are satisfied as long as $E[X_i^2] < \infty$.

\section{Main Results}\label{sec:main}

\subsection{Asymptotic Behavior of the Wald Estimator}\label{subsec:MPIV}
In this section, we study the asymptotic properties of the standard Wald estimator (i.e., the two-stage least squares estimator of $Y_i$ on $D_i$ using $A_i$ as an instrument) of $\Delta(Q)$ under a matched pairs design. In order to introduce this estimator, define
\begin{equation}\label{eq:meanYa}
\begin{split}
\hat{\psi}_n(a) = \frac{1}{n} \sum_{1\leq i\leq 2n: A_i = a} Y_i~,
\end{split}
\end{equation}
\begin{equation}\label{eq:meanDa}
\begin{split}
\hat{\phi}_n(a) = \frac{1}{n} \sum_{1\leq i\leq 2n: A_i = a} D_i~.
\end{split}
\end{equation}
Using this notation, the Wald estimator is defined as
\begin{equation}\label{eq:hatDelta}
\begin{split}
\hat{\Delta}_n = \frac{\hat{\psi}_n(1) - \hat{\psi}_n(0)}{\hat{\phi}_n(1) - \hat{\phi}_n(0)}~.
\end{split}
\end{equation}
Note that this estimator may be obtained as the ratio of the estimator of the coefficient of $A_i$ in an ordinary least squares regression of $Y_i$ on a constant and $A_i$ (the ``reduced form") to the estimator of the coefficient of $A_i$ in an ordinary least squares regression of $D_i$ on a constant and $A_i$ (the ``first stage"). Theorem \ref{theorem:main} establishes the limiting distribution of $\hat{\Delta}_n$ under a matched pairs design. 
\begin{theorem}\label{theorem:main}
Suppose $Q$ satisfies Assumption \ref{ass:Q} and the treatment assignment mechanism satisfies Assumptions \ref{ass:assignment}--\ref{ass:pairs}.  Then,
\begin{align*}
\sqrt{n} \left( 
\hat{\Delta}_n
-
\Delta(Q)
\right) \cd 
N\left(0, \nu^2   \right),
\end{align*}
where
\begin{multline*}
\nu^2 = \frac{1}{P \{C_i = 1\}^2}\Biggr (E[\var[Y^\ast_i(1)|X_i]] + E[\var[Y^\ast_i(0)|X_i]] \\
+ \frac{1}{2}E\left[((E[Y^\ast_i(1)|X_i] - E[Y^\ast_i(1)]) - (E[Y^\ast_i(0)|X_i] - E[Y^\ast_i(0)]))^2\right]\Biggr )~,
\end{multline*}
with
\begin{equation}\label{eq:Ystar-a}
Y^\ast_i(a) = \tilde Y_i(a) - \Delta(Q)D_i(a)~,
\end{equation}
for $a \in \{0, 1\}$.

\end{theorem}

We derive Theorem \ref{theorem:main} by reproducing arguments in Lemma S.1.4 in \cite{bai2022mp}, but replacing the usual potential outcomes $Y_i(a)$ with the transformed outcomes $Y_i^\ast(a)$\footnote{We note that although a version of Theorem \ref{theorem:main} can be obtained by directly appealing to more recent results in \cite{bai2023efficient}, the regularity conditions introduced in this paper are in fact weaker and, in our view, easier to interpret than the general conditions considered there.}. As a result, the the numerator of $\nu^2$ corresponds exactly to the limiting variance obtained in \cite{bai2022mp} with the usual potential outcomes $Y_i(a)$ replaced with $Y_i^\ast(a)$.  In particular, when there is perfect compliance, so that $D_i = A_i$, $P \{C_i = 1\} = 1$, and $D_i(a) = a$, the limiting variance we obtain in Theorem \ref{theorem:main} corresponds exactly to the limiting variance derived in \cite{bai2022mp}.  It can be shown that our expression for $\nu^2$ attains the efficiency bound derived in \cite{bai2023efficient} over a broad class of treatment assignments which include matched pairs as a special case \citep[it can also be shown that this bound coincides with the bounds derived in][in settings  with observational data with i.i.d.\ assignment when $P\{A_i = a|X_i\} = 1/2$, where the local average treatment effect is estimated non-parametrically using the covariates $X_i$: see Lemma \ref{lemma:efficiency-bound} for details.]{frolich2007nonparametric, hong2010semiparametric} 


\subsection{Variance Estimation}\label{sec:variance}


In this section, we construct a consistent variance estimator for the limiting variance $\nu^2$, and then contrast this to the asymptotic behavior of standard regression-based variance estimators. As noted in the discussion following Theorem \ref{theorem:main}, the expression for $\nu^2$ corresponds exactly to the limiting variance obtained in \cite{bai2022mp} with the usual potential outcomes $Y_i(a)$ replaced with the transformed outcomes $Y_i^\ast(a)$. We thus follow the variance construction from \cite{bai2022mp}, but with a feasible version of $Y_i^\ast(a)$ defined as
\begin{equation} \label{eq:transformed}
\hat{Y}_i = Y_i - \hat{\Delta}_nD_i~.
\end{equation}
This strategy leads to the following variance estimator:
\begin{equation} \label{eq:mpiv-varest}
\hat{\nu}^2_n = \frac{\hat{\tau}^2_n - \frac{1}{2}(\hat{\lambda}^2_n + \hat{\Gamma}^2_n)}{\left(\hat{\phi}_n(1) - \hat{\phi}_n(0)\right)^2}~,   
\end{equation}
where
\begin{align*}
\hat{\tau}^2_n & = \frac{1}{n}\sum_{1 \le j \le n}(\hat{Y}_{\pi(2j)} - \hat{Y}_{\pi(2j-1)})^2 \\
\hat{\lambda}^2_n & = \frac{2}{n}\sum_{1 \le j \le \lfloor \frac{n}{2} \rfloor}\left(\hat{Y}_{\pi(4j-3)} - \hat{Y}_{\pi(4j-2)}\right)\left(\hat{Y}_{\pi(4j-1)} - \hat{Y}_{\pi(4j)}\right)\left(A_{\pi(4j-3)} - A_{\pi(4j-2)}\right)\left(A_{\pi(4j-1)} - A_{\pi(4j)}\right) \\
\hat{\Gamma}_n & = \frac{1}{n}\sum_{1 \le i \le 2n: A_i = 1}\hat{Y}_i - \frac{1}{n}\sum_{1 \le i \le 2n: A_i = 0}\hat{Y}_i~.
\end{align*}
Note that the construction of the numerator of $\hat{\nu}^2_n$ can be motivated using a similar intuition to what has been previously discussed in \cite{bai2022mp}: to consistently estimate
\[ E[(E[Y^\ast_i(1)|X_i] - E[Y^\ast_i(0)|X_i])^2]~, \]
ideally we would like access to four different units with similar values of $X_i$, of which two are treated. However, because each pair only contains two units, we need to average across ``pairs of pairs" of units, where two pairs are grouped together so that they are ``close'' in terms of $X_i$. We establish the following consistency result for $\hat{\nu}^2_n$:

\begin{theorem}\label{theorem:varianceestimation}
Suppose $Q$ satisfies Assumption \ref{ass:Q} and the treatment assignment mechanism satisfies Assumptions \ref{ass:assignment}--\ref{ass:pairsofpairs}. Then,
\begin{align*}
\hat{\nu}_n^2 \cp \nu^2.
\end{align*}
\end{theorem}


Assumption \ref{ass:Q}(a) implies that $\nu^2 > 0$. We therefore immediately obtain the following corollary which establishes the asymptotic exactness of a $t$-test for the null hypothesis \eqref{eq:H0} constructed using the variance estimator $\hat{\nu}^2_n$:
\begin{corollary}\label{cor:exact_test}
Suppose $Q$ satisfies Assumption \ref{ass:Q} and the treatment assignment mechanism satisfies Assumptions \ref{ass:assignment}--\ref{ass:pairsofpairs}. Then,
\begin{align*}
\frac{\sqrt{n} \left( \hat{\Delta}_n - \Delta(Q) \right)}{\hat \nu_n} \cd N\left( 0, 1\right)~.
\end{align*}
\end{corollary}

Next, we consider the limiting behavior of the usual heteroskedasticity-robust variance estimator obtained from a two-stage least squares regression of $Y_i$ on a constant and $D_i$, using $A_i$ as an instrument, which we denote by $\hat \omega_n^2$. Theorem \ref{theorem:robust} derives the limit in probability of $\hat \omega_n^2$.

\begin{theorem}\label{theorem:robust}
Suppose $Q$ satisfies Assumption \ref{ass:Q} and the treatment assignment mechanism satisfies Assumptions \ref{ass:assignment}--\ref{ass:pairs}.  Then,
\begin{align*}
\hat{\omega}_{n}^2 \cp \omega^2~,
\end{align*}
where 
\[\omega^2 = \frac{1}{P \{C_i = 1\}^2} (\var[Y_i^\ast(1)] + \var[Y_i^\ast(0)])~.\]
In particular, $\omega^2 \ge \nu^2$ and the inequality is strict unless
\[E[Y^*_i(1) + Y^*_i(0)|X_i] = E[Y^*_i(1) + Y^*_i(0)]\]
with probability one under $Q$.
\end{theorem}

Finally, we consider the heteroskedasticity-robust variance estimator obtained from a two-stage least squares regression with pair fixed effects, which is a specification commonly used in practice: see \cite{groh2016macroinsurance}, and \cite{resnjanskij2021can} for examples. Specifically, let $\hat \alpha_n$ be the two-stage least squares estimator for $\alpha$ in the linear regression
\begin{equation} \label{eq:pfe-noadj}
Y_i = \alpha D_i + \sum_{1 \leq j \leq n} \theta_j I \{i \in \{\pi(2j - 1), \pi(2j)\}\} + \epsilon_i
\end{equation}
using $A_i$ as an instrument for $D_i$. Let $\hat \omega^2_{n,\rm pfe}$ denote the usual heteroskedasticity-robust variance estimator (HC0) for $\hat \alpha_n$ and let $\hat \omega^2_{n,\rm pfe, HC1}$ denote the robust variance estimator with a degree-of-freedom adjustment (HC1). In this particular case,
\begin{equation}\label{eq:pfe-HC1}
\hat \omega^2_{n,\rm pfe, HC1} = \frac{2n}{2n - (n + 1)} \hat \omega^2_{n,\rm pfe}~, 
\end{equation}
because the number of regressors in \eqref{eq:pfe-noadj} is $n + 1$.

\begin{theorem} \label{thm:hc1-fe}
Suppose $Q$ satisfies Assumption \ref{ass:Q} and the treatment assignment mechanism satisfies Assumptions \ref{ass:assignment}--\ref{ass:pairs}. Then,
\begin{align*}
\hat \omega^2_{n, \rm pfe} & \stackrel{P}{\to} \frac{1}{2} \omega_{\rm pfe, HC1}^2  \\
\hat \omega^2_{n, \rm pfe, HC1}  & \stackrel{P}{\to} \omega_{\rm pfe, HC1}^2~,
\end{align*}
where
\[ \omega_{ \rm pfe, HC1}^2 = \frac{1}{P \{C_i = 1\}^2} (E[\var[Y_i^\ast(1) | X_i]] + E[\var[Y_i^\ast(0) | X_i]] + E[(E[Y^\ast_i(1)|X_i] - E[Y^\ast_i(0)|X_i])^2])~. \]
In particular, $\omega_{ \rm pfe, HC1}^2 \ge \nu^2$ and the inequality is strict unless \[E[Y^\ast_i(1) - Y^\ast_i(0)|X_i] = 0\] with probability one under $Q$.
\end{theorem}

From Theorems \ref{theorem:robust} and \ref{thm:hc1-fe} we obtain that neither cluster-robust standard error is consistent for $\nu^2$ unless the baseline covariates are irrelevant for the transformed potential outcomes $Y_i^*(a)$ in an appropriate sense. In Section \ref{sec:sims}, we illustrate that this conservativeness translates into a lack of power relative to tests constructed using our consistent variance estimator $\hat{\nu}_n^2$.

\begin{remark}
Note that Theorem \ref{thm:hc1-fe} further reveals that the limit in probability of $\hat \omega^2_{n, \rm pfe}$ may in general be strictly smaller than $\nu^2$, and therefore tests constructed using this estimator are not guaranteed to control size in general. See \cite{bai2023inference} for similar findings in the case of perfect compliance.
\end{remark}

\section{Covariate Adjustment} \label{sec:adjustment}
In this section, we consider a generalization of the estimator $\hat \Delta_n$ defined in Section \ref{subsec:MPIV} that allows for covariate adjustment using the additional, observed, baseline covariates $W^{(n)}$ that were not used when forming pairs. In Section \ref{sec:gen_adjust}, we derive general results on covariate adjustment using arbitrarily-specified working models of the conditional expectations of the outcome and treatment take-up with respect to the covariates. In Section \ref{sec:linear_adjust}, we show how a careful application of these earlier results using an optimal linear covariate-adjusted estimator can ensure an improvement over the unadjusted Wald estimator $\hat{\Delta}_n$ in terms of precision. 

\subsection{General Covariate Adjustments}\label{sec:gen_adjust}
Following \cite{bai2023covariate}, note that it can be shown under Assumption \ref{ass:assignment} that for any $a \in \{0, 1 \}$, $m_{a, \tilde Y}: \mathbb{R}^{k_x} \times \mathbb{R}^{k_w} \to \mathbb{R}$, and $m_{a, D}: \mathbb{R}^{k_x} \times \mathbb{R}^{k_w} \to \mathbb{R}$ such that $E[|m_{a, \tilde Y}(X_i, W_i)|] < \infty$, $E[|m_{a, D}(X_i, W_i)|] < \infty$,
\begin{equation}\label{eq:moment-tildeY}
E[ 2 I \{ A_i = a \} ( Y_i - m_{a, \tilde Y}( X_i, W_i ) ) + m_{a, \tilde Y}( X_i, W_i ) ] = E[ \tilde Y_i(a) ],
\end{equation}
\begin{equation}\label{eq:moment-D}
E\left[ 2 I \{ A_i = a \} ( D_i - m_{a, D}( X_i, W_i) ) + m_{a, D}( X_i, W_i ) \right] = E\left[ D_i(a) \right].
\end{equation}
Here, we view $m_{a, D}$ and $m_{a, \tilde Y}$ as ``working models'' of $E[\tilde{Y}_i(a)|X_i,W_i]$ and $E[D_i(a)|X_i,W_i]$, respectively, but we emphasize that \eqref{eq:moment-tildeY}--\eqref{eq:moment-D} hold even if these models are incorrectly specified.
These two moment conditions motivate a covariate-adjusted estimator defined as
\[ \hat \Delta_n^{\rm adj} = \frac{\hat \psi_n^{\rm adj}(1) - \hat \psi_n^{\rm adj}(0)}{\hat \phi_n^{\rm adj}(1) - \hat \phi_n^{\rm adj}(0)}~, \]
where
\begin{align*}
\hat \psi_n^{\rm adj}(a) & = \frac{1}{2n} \sum_{1 \leq i \leq 2n} (2 I \{A_i = a\} (Y_i - \hat m_{a, \tilde Y}(X_i, W_i)) + \hat m_{a, \tilde Y}(X_i, W_i))~, \\
\hat \phi_n^{\rm adj}(a) & = \frac{1}{2n} \sum_{1 \leq i \leq 2n} (2 I \{A_i = a\} (D_i - \hat m_{a, D}(X_i, W_i)) + \hat m_{a, D}(X_i, W_i))~,
\end{align*}
and $\hat m_{a, \tilde Y}$ and $\hat m_{a, D}$ are suitable estimators of $m_{a, \tilde Y}$ and $m_{a, D}$. Note that if we set $\hat m_{a,\tilde{Y}} = \hat m_{a, D} = 0$, then $\hat{\Delta}_n^{\rm adj}$ simplifies to $\hat{\Delta}_n$. Let \[m_{a, \tilde Y D}(X_i,W_i) = m_{a, \tilde Y}(X_i,W_i) - \Delta(Q) m_{a, D}(X_i,W_i)~.\] Theorem \ref{theorem:adjustment} below establishes the limiting distribution of $\hat{\Delta}_n^{\rm adj}$ for a matched-pairs design under the following high-level assumption on the working models:

\begin{assumption}\label{ass:m}
The functions $m_{a, \tilde Y D}$ for $a \in \{0, 1 \}$ satisfy \vspace{-.25cm}
\begin{enumerate}[\rm (a)]

        \item $E[m_{a,\tilde Y D}(X_i,W_i)^2] < \infty$ for $a \in \{0, 1\}$.

        \item $E[m_{a, \tilde Y D}(X_i, W_i) | X_i = x]$, $E[(m_{a, \tilde Y D}(X_i, W_i))^2 | X_i = x]$, and $E[m_{a, \tilde Y D}(X_i, W_i) Y_i^\ast(a) | X_i = x]$ for $a \in \{0, 1 \}$, and $E[m_{1, \tilde Y D}(X_i, W_i) m_{0, \tilde Y D}(X_i, W_i) | X_i = x]$ are Lipschitz.
    \end{enumerate}
    \end{assumption}

    \begin{theorem}\label{theorem:adjustment}
    Suppose $Q$ satisfies Assumption \ref{ass:Q}, the treatment assignment mechanism satisfies Assumptions \ref{ass:assignment}--\ref{ass:pairs}, and $m_{a, \tilde Y}$ and $m_{a, D}$ for $a\in \{0, 1 \}$ satisfy Assumption \ref{ass:m}. 
    Further suppose $\hat m_{a, \tilde Y}$ and $\hat m_{a, D}$ satisfy that 
    \begin{equation}\label{eq:m-Y-hat-close}
    \frac{1}{\sqrt{2n}} \sum_{1\leq i\leq 2n} (2A_i - 1) (\hat m_{a, \tilde Y}(X_i, W_i) - m_{a, \tilde Y}(X_i, W_i)) \cp 0~,
    \end{equation}
    \begin{equation}\label{eq:m-D-close}
    \frac{1}{\sqrt{2n}} \sum_{1\leq i\leq 2n} (2A_i - 1) (\hat m_{a, D}(X_i, W_i) - m_{a, D}(X_i, W_i)) \cp 0~.
    \end{equation}
    Then, $\hat \Delta_n^{\rm adj}$ satisfies
    \begin{align*}
    \sqrt{n} \left( \hat \Delta_n^{\rm adj}
    -
    \Delta(Q)
    \right) \cd 
    N\left(0, \nu_{\rm adj}^2  \right),
    \end{align*}
    where $\nu_{\rm adj}^2 = \frac{1}{P \{C_i = 1\}^2} ( \nu_{1, \rm adj}^2 + \nu_{2, \rm adj}^2 + \nu_{3, \rm adj}^2 )$ with
    \begin{align*}
    \nu_{1, \rm adj}^2 
    &= \frac{1}{2} E[ \var[ E\left[ Y_i^\ast(1) + Y_i^\ast(0) | X_i, W_i \right] 
    - ( m_{1, \tilde Y D}(X_i, W_i) + m_{0, \tilde Y D}(X_i, W_i) ) | X_i ] ]~, \\
    \nu_{2, \rm adj}^2 
    &= \frac{1}{2} \var[ E[ Y^\ast_i(1) - Y^\ast_i(0) | X_i, W_i] ]~, \\
    \nu_{3, \rm adj}^2
    &= E[ \var[ Y^\ast_i(1) | X_i, W_i ] + \var[ Y^\ast_i(0) | X_i, W_i ] ]~.
    \end{align*}
    \end{theorem}
       In Section \ref{sec:linear_adjust} we provide low-level sufficient conditions for Assumptions \ref{ass:m} and \eqref{eq:m-Y-hat-close}-\eqref{eq:m-D-close} for the special case in which $m_{a,\tilde{Y}}$ and $m_{a,D}$ correspond to linear working models that are optimal in the sense of minimizing the limiting variance of $\hat \Delta_n^{\rm adj}$ among all linear working models. We derive Theorem \ref{theorem:adjustment} by reproducing arguments in Theorem 3.1 in \cite{bai2023covariate}, but replacing the working model $m_a$ with the transformed working model $m_{a, \tilde YD}$. As a result, the then numerator $\nu_{\rm adj}^2$  corresponds exactly to the limiting variance obtained in \cite{bai2023covariate} with the working models in \cite{bai2023covariate} replaced with $m_{a, \tilde YD}$. As expected, $\nu^2_{\rm adj} = \nu^2$ when the working models are set to zero. In general, $\nu^2_{\rm adj}$ is not guaranteed to be weakly smaller than $\nu^2$ for all choices of working models, but we note that $\nu_{\rm adj}^2$ is minimized (and thus smaller than $\nu^2$) when $\nu_{1, \rm adj}^2 = 0$, i.e., when the working models satisfy
    \begin{multline*}
    E[Y_i^\ast(1) + Y_i^\ast(0) | X_i, W_i] - E[Y_i^\ast(1) + Y_i^\ast(0) | X_i] \\
    = m_{1, \tilde Y D}(X_i, W_i) + m_{0, \tilde Y D}(X_i, W_i) - E[m_{1, \tilde Y D}(X_i, W_i) + m_{0, \tilde Y D}(X_i, W_i) | X_i]
    \end{multline*}
    with probability one. This property holds, in particular, when $m_{a,\tilde{Y}}$ and $m_{a,D}$ are correctly specified. Moreover, under correct specification, it can be shown that $\nu^2_{\rm adj}$ coincides with the efficiency bound derived in \cite{bai2023efficient}.
    

    Next, we construct a consistent variance estimator for the limiting variance $\nu_{\rm adj}^2$. As noted in the discussion following Theorem \ref{theorem:adjustment}, the numerator for $\nu_{\rm adj}^2$ corresponds exactly to the limiting variance obtained in \cite{bai2023covariate} with the usual potential outcomes $Y_i(a)$ replaced with the transformed outcomes $Y_i^\ast(a)$ and working models in  \cite{bai2023covariate} replaced with $m_{a, \tilde YD}$. We thus follow the variance construction from \cite{bai2023covariate}, but with a feasible version of $Y_i^\ast(a)$ defined as
    \[\hat{Y}_i = Y_i - \hat{\Delta}_nD_i~,\]
    and a feasible version of $m_{a, \tilde{Y}D}(X_i,W_i) = m_{a, \tilde Y}(X_i, W_i) - \Delta(Q) m_{a, D}(X_i, W_i)$ defined as 
    \[
    \hat m_{a, \tilde Y D}(X_i, W_i) = \hat m_{a, \tilde Y}(X_i, W_i) - \hat{\Delta}_n \hat m_{a, D}(X_i, W_i)~.
    \]
    This strategy leads to the following variance estimator:
    \begin{equation} \label{eq:adj-varest}
    \hat \nu_{n, \rm adj}^2 = \frac{\hat \tau_{n, \rm adj}^2 - \frac{1}{2} ( \hat \lambda_{n, \rm adj} +  \hat \Gamma_{n, \rm adj}^2 )}{(\hat \phi_n^{\rm adj}(1) - \hat \phi_n^{\rm adj}(0))^2}~,
    \end{equation}
    where
    \begin{align*}
    \hat \tau_{n, \rm adj}^{2} &= \frac{1}{n} \sum_{1 \leq j \leq n} ( \hat Y_{\pi(2j-1), \rm adj} - \hat Y_{\pi(2j), \rm adj} )^2~, \\
    \hat \lambda_{n, \rm adj} &= \frac{2}{n} \sum_{1 \leq j \leq \lfloor \frac{n}{2} \rfloor} ( \hat Y_{\pi(4j-3), \rm adj} - \hat Y_{\pi(4j-2), \rm adj} ) \\
    & \hspace{3em} \times ( \hat Y_{\pi(4j-1), \rm adj} - \hat Y_{\pi(4j), \rm adj} ) ( A_{\pi(4j-3)} - A_{\pi(4j-2)} ) ( A_{\pi(4j-1)} - A_{\pi(4j)} )~, \\
    \hat \Gamma_{n, \rm adj} &= \frac{1}{n}\sum_{1 \le i \le n: A_i = 1} \hat Y_{i, \rm adj} - \frac{1}{n}\sum_{1 \le i \le n: A_i = 0} \hat Y_{i, \rm adj}~, \\
    \hat Y_{i, \rm adj} &= \hat Y_i - \frac{1}{2} ( \hat m_{1, \tilde Y D}(X_i, W_i) + \hat m_{0, \tilde Y D}(X_i, W_i)  )~.
    \end{align*}
    The following theorem establishes the consistency of $\hat{\nu}_{n, \rm adj}^2$ for $ \nu_{\rm adj}^2$:
    \begin{theorem}\label{theorem:varianceestimation-adj}
    Suppose $Q$ satisfies Assumption \ref{ass:Q}, the treatment assignment mechanism satisfies Assumptions \ref{ass:assignment}--\ref{ass:pairsofpairs}, and $m_{a, \tilde Y}$ and $m_{a, D}$ for $a\in \{0, 1 \}$ satisfy Assumption \ref{ass:m}. Further suppose $\hat m_{a, \tilde Y}$ and $\hat m_{a, D}$ satisfy (\ref{eq:m-Y-hat-close})--(\ref{eq:m-D-close}) and 
    \begin{equation}\label{eq:m-Y-hat-close-adj}
    \frac{1}{2n} \sum_{1\leq i\leq 2n} ( \hat m_{a, \tilde Y}(X_i, W_i) - m_{a, \tilde Y}(X_i, W_i) )^2 \cp 0~,
    \end{equation}
    \begin{equation}\label{eq:m-D-close-adj}
    \frac{1}{2n} \sum_{1\leq i\leq 2n} ( \hat m_{a, D}(X_i, W_i) - m_{a, D}(X_i, W_i) )^2 \cp 0~.
    \end{equation}
    Then,
    \begin{align*}
    \hat{\nu}_{n, \rm adj}^2 \cp \nu_{\rm adj}^2~.
    \end{align*}
    \end{theorem}

    \subsection{Optimal Linear Adjustments}\label{sec:linear_adjust}
    We now consider a simple setting of practical interest which specializes Theorem \ref{theorem:adjustment} to the case in which $m_{a,\tilde{Y}}$ and $m_{a,D}$ are the optimal linear working models in a sense to be made formal below.
    To that end, let $\zeta_i = \zeta(X_i, W_i)$ be a user-specified transformation of the baseline characteristics $(X_i, W_i)$ given by some function $\zeta: \mathbb R^{k_x} \times \mathbb R^{k_w} \to \mathbb R^p$. Let $\hat m_{a, \tilde Y}(X_i, W_i) = \zeta_i' \hat \beta_n^Y$ and $\hat m_{a, D}(X_i, W_i) = \zeta_i' \hat \beta_n^D$ for $a \in \{0, 1\}$, where $\hat \beta_n^Y$ and $\hat \beta_n^D$ are the ordinary least squares estimators of $\beta^Y$ and $\beta^D$ in the following two linear regressions with pair fixed effects:
    \begin{align}
    \label{eq:pfe-Y} Y_i & = \alpha^Y A_i + \zeta_i' \beta^Y + \sum_{1 \leq j \leq n} \theta_j^Y I \{i \in \{\pi(2j - 1), \pi(2j)\}\} + \epsilon_i^Y \\
    \label{eq:pfe-D} D_i & = \alpha^D A_i + \zeta_i' \beta^D + \sum_{1 \leq j \leq n} \theta_j^D I \{i \in \{\pi(2j - 1), \pi(2j)\}\} + \epsilon_i^D~.
    \end{align}
    Lemma \ref{lem:FE_TSLS} in the appendix shows that the adjusted estimator $\hat{\Delta}_n^{\rm adj}$ obtained from this choice of $\hat m_{a, \tilde Y}$ and  $\hat m_{a, D}$ could alternatively be computed as the regression coefficient $\hat{\alpha}_n^{\rm IV}$ in the following two-stage least squares regression with $A_i$ as an instrument for $D_i$:
    \begin{equation} \label{eq:pfe-iv}
    Y_i = \alpha D_i + \zeta_i' \beta + \sum_{1 \leq j \leq n} \theta_j I \{i \in \{\pi(2j - 1), \pi(2j)\}\} + \epsilon_i~.
    \end{equation}
    Versions of this estimator have been used in, for instance, \cite{groh2016macroinsurance} and \cite{resnjanskij2021can}. It is also a natural counterpart for the OLS estimator with pair fixed effects which is widely used in settings with perfect compliance; see, for instance, \cite{glennerster2014running}.
    In order to analyze the large-sample behavior of this estimator, we introduce the following assumption:
    \begin{assumption} \label{ass:adjustment}
    The function $\zeta(\cdot)$ is such that
    \begin{enumerate}[\rm (a)]
    \item No component of $\zeta_i$ is a constant and $E[\var[\zeta_i | X_i]]$ is nonsingular.
    \item $\var[\zeta_i] < \infty$.
    \item $E[\zeta_i | X_i = x]$, $E[\zeta_i \zeta_i' | X_i = x]$, $E[\zeta_i \tilde Y_i(a) | X_i = x]$, $E[\zeta_i D_i(a) | X_i = x]$ are Lipschitz.
\end{enumerate}
\end{assumption}
Theorem \ref{thm:zeta} shows that Assumption \ref{ass:adjustment} provides low-level sufficient conditions for Assumption \ref{ass:m} and \eqref{eq:m-Y-hat-close}--\eqref{eq:m-D-close}, and further verifies the optimality of these working models among all linear working models:

\begin{theorem} \label{thm:zeta}
Suppose $Q$ satisfies Assumption \ref{ass:Q}, the treatment assignment mechanism satisfies Assumptions \ref{ass:assignment}--\ref{ass:pairs}, and in addition Assumption \ref{ass:adjustment} is satisfied. Then, 
\begin{align*}
\hat \beta_n^Y & \stackrel{P}{\to} \beta^{\tilde Y} = (2 E[\var[\zeta_i | X_i]])^{-1} E[\cov[\zeta_i, \tilde Y_i(1) + \tilde Y_i(0) | X_i]] \\
\hat \beta_n^D & \stackrel{P}{\to} \beta^D = (2 E[\var[\zeta_i | X_i]])^{-1} E[\cov[\zeta_i, D_i(1) + D_i(0) | X_i]]~.
\end{align*}
In addition, \eqref{eq:m-Y-hat-close}--\eqref{eq:m-D-close} and Assumption \ref{ass:m} are satisfied for $\hat m_{a, \tilde Y}(X_i, W_i) = \zeta_i' \hat \beta_n^Y$, $\hat m_{a, D}(X_i, W_i) = \zeta_i' \hat \beta_n^D$, $m_{1, \tilde Y}(X_i, W_i) = m_{0, \tilde Y}(X_i, W_i) = \zeta_i' \beta^{\tilde Y}$, and $m_{1, \tilde D}(X_i, W_i) = m_{0, \tilde D}(X_i, W_i) = \zeta_i' \beta^D$. Moreover, $\hat \Delta_n^{\rm adj}$ is optimal in the sense of minimizing $\nu_{\rm adj}^2$ among all choices of $m_{a,D}$ and $m_{a,\tilde Y}$ that are linear in $\zeta_i$.
\end{theorem}
We conclude by noting that, since the unadjusted estimator $\hat{\Delta}_n$ is contained in the class of covariate adjusted estimators for which $m_{a,D}$ and $m_{a,\tilde Y}$ are linear in $\zeta_i$ (by setting both working models to zero), Theorem \ref{thm:zeta} establishes that the optimal linear adjusted estimator is guaranteed to be (weakly) more efficient than the unadjusted Wald estimator, regardless of whether or not the linear working models are correctly specified.


\begin{remark}\label{rem:FE_TSLS}
Although $\hat{\Delta}^{\rm adj}_n$ can be obtained using $\hat{\alpha}^{\rm IV}_n$ in this case, we emphasize that the usual heteroskedasticity-robust estimator of the variance of $\hat{\alpha}_n^{\rm IV}$ is not consistent for the limiting variance derived in Theorem \ref{theorem:adjustment}. We therefore recommend that practitioners instead use the consistent estimator presented in Theorem \ref{theorem:varianceestimation-adj}.
\end{remark}

\section{Simulations}\label{sec:sims}
\subsection{No Covariate Adjustment}\label{sec:sims_no}
In this section, we examine the finite-sample behavior of the estimation and inference procedures introduced in Section \ref{sec:main}. Following the simulation designs in \cite{bai2022mp}, the potential outcomes and take-up decisions are generated according to the equations:
\begin{align*}
    Y_i(d) &= \mu_d + m_d(X_i) + \sigma_d(X_i) \epsilon_{d, i} \\
    D_i(0) &= I \{ 0.2 X_i > \epsilon_{3,i}\} \\
    D_i(1) &= 
    \begin{cases}
        I \{0.5 + 0.2 X_i >  \epsilon_{4, i}\} & \text{if } D_i(0) = 0 \\
        1 & \text{otherwise}~,
    \end{cases}
\end{align*}
where $\mu_d$, $m_d(X_i)$, $\sigma_d(X_i)$, and $\epsilon_{d, i}$ are specified in each model below. We consider the following model specifications: 

\begin{itemize}
    \item []
    \begin{itemize}
        \item []
        \begin{itemize}
            \item [\textbf{\textbf{Model 1}:}]  $X_i \sim \text{Unif}[0, 1]$; $m_0(X_i) = m_1(X_i) = \left(X_i - \frac{1}{2}  \right)$; $\epsilon_{0, i}, \epsilon_{1,i} \sim N(0, 1)$; $\epsilon_{3, i}, \epsilon_{4, i} \sim \text{Unif}[0, 1]$; $\sigma_0(X_i) = \sigma_1(X_i) = 1$; $\mu_0 = 0; \mu_1 \in [-1, 1]$.
            
            \item [\textbf{\textbf{Model 2}:}] As in Model 1, but $m_0(X_i) = 0$, $m_1(X_i) = 10 \left(X_i^2 - \frac{1}{3} \right)$.

            \item [\textbf{\textbf{Model 3}:}] As in Model 2, but $\sigma_0(X_i) = \sigma_1(X_i) = X_i^2$.
        \end{itemize}
    \end{itemize}
\end{itemize}

For each model, let $\Delta_0$ denote the value of the LATE when $\mu_1 = 0$.\footnote{For each model, $\Delta_0$ is computed numerically with a large sample. The values are -0.0000203726,  0.0859858425,  0.0903371248 for Models 1-3.} Because $\dim(X_i) = 1$, we construct pairs by sorting units according to $X_i$ and matching adjacent units. Table \ref{table:rej-prob-models123} reports the rejection probabilities from 5,000 Monte Carlo replications for testing the null hypothesis \eqref{eq:H0} against the alternative implied by setting $\mu_1 = 1/2$,  using $t$-tests constructed using either the regression-based variance $\hat{\omega}_n^2$, the regression-based variance with pair fixed effects and HC1 correction $\hat{\omega}_{n, \rm pfe, HC1}^2$, or the consistent estimator $\hat{\nu}^2_n$. As expected given our theoretical results, tests based on $\hat \omega_n^2$ and $\hat{\omega}_{n, \rm pfe, HC1}^2$ are generally conservative, which leads to a loss of power under the alternative relative to the test based on our consistent variance estimator $\hat \nu^2_n$. Figure \ref{figure:power-curve} displays the power curves for testing the null hypothesis \eqref{eq:H0} for Models 1--3 with $\mu_1 \in [-1, 1]$ and $2n = 200$. We emphasize that, although $\hat{\omega}_{n, \rm pfe, HC1}^2$ is less conservative than  $\hat{\omega}_n^2$ for these simulation designs, our theoretical results in Section \ref{sec:variance} establish that this is not guaranteed to be the case in general.

\begin{table}[ht!]
    \centering
    \begin{tabular}{ c c  c c c c c c    }
         \toprule
         \multicolumn{1}{c}{} & 
         \multicolumn{1}{c}{} &
         \multicolumn{3}{c}{$H_0: \Delta = \Delta_0$} & \multicolumn{3}{c}{$H_1: \Delta = \Delta_0 + \frac{1}{2}$}  \\
         \cmidrule(lr){3-5} \cmidrule(lr){6-8} 
         Model & Sample Size & $\hat{\omega}_n^2$   & $\hat{\omega}_{n, \rm pfe, HC1}^2$ & $\hat \nu^2_n$ & $\hat{\omega}_n^2$  & $\hat{\omega}_{n, \rm pfe, HC1}^2$ & $\hat \nu^2_n$   \\
         \midrule
        \multirow{4}*{1} & 200 & 3.86 & 4.88 & 4.98 & 44.60 & 47.48 & 47.98\\
         & 800 & 4.10 & 4.84 & 4.96 & 95.84 & 96.42 & 96.48\\
         & 1600 & 3.92 & 4.72 & 4.78 & 99.84 & 99.84 & 99.84\\
         & 3200 & 4.40 & 5.34 & 5.34 & 100.00 & 100.00 & 100.00\\
        \addlinespace
        \multirow{4}*{2} & 200 & 1.72 & 3.12 & 4.60 & 10.92 & 13.86 & 19.94\\
         & 800 & 1.88 & 3.06 & 4.92 & 43.94 & 52.52 & 59.44\\
         & 1600 & 1.72 & 2.98 & 4.86 & 76.60 & 82.44 & 87.26\\
         & 3200 & 1.76 & 3.16 & 5.16 & 97.66 & 98.60 & 99.24\\
        \addlinespace
        \multirow{4}*{3} & 200 & 1.36 & 2.62 & 4.76 & 11.16 & 15.38 & 24.10\\
         & 800 & 1.38 & 2.46 & 5.00 & 51.72 & 63.00 & 71.76\\
         & 1600 & 1.12 & 2.34 & 4.78 & 85.48 & 91.34 & 94.64\\
         & 3200 & 1.26 & 2.40 & 4.80 & 99.38 & 99.68 & 99.86\\
        \bottomrule
        \end{tabular}
    \caption{Rejection rates for Models 1, 2, 3 without covariate adjustment} 
    \label{table:rej-prob-models123}
\end{table}



\begin{figure}[ht!]

\centering
\includegraphics[width=.5\textwidth]{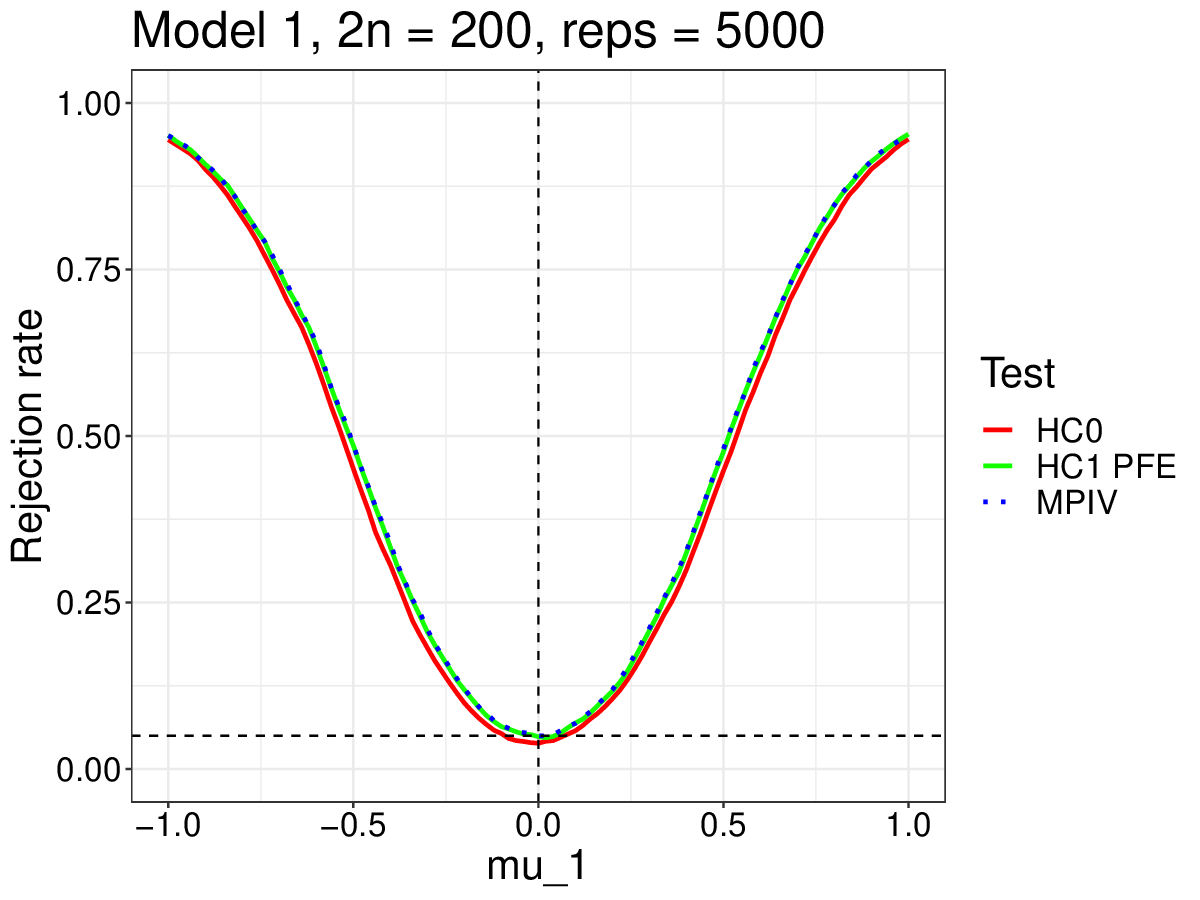}\hfill
\includegraphics[width=.5\textwidth]{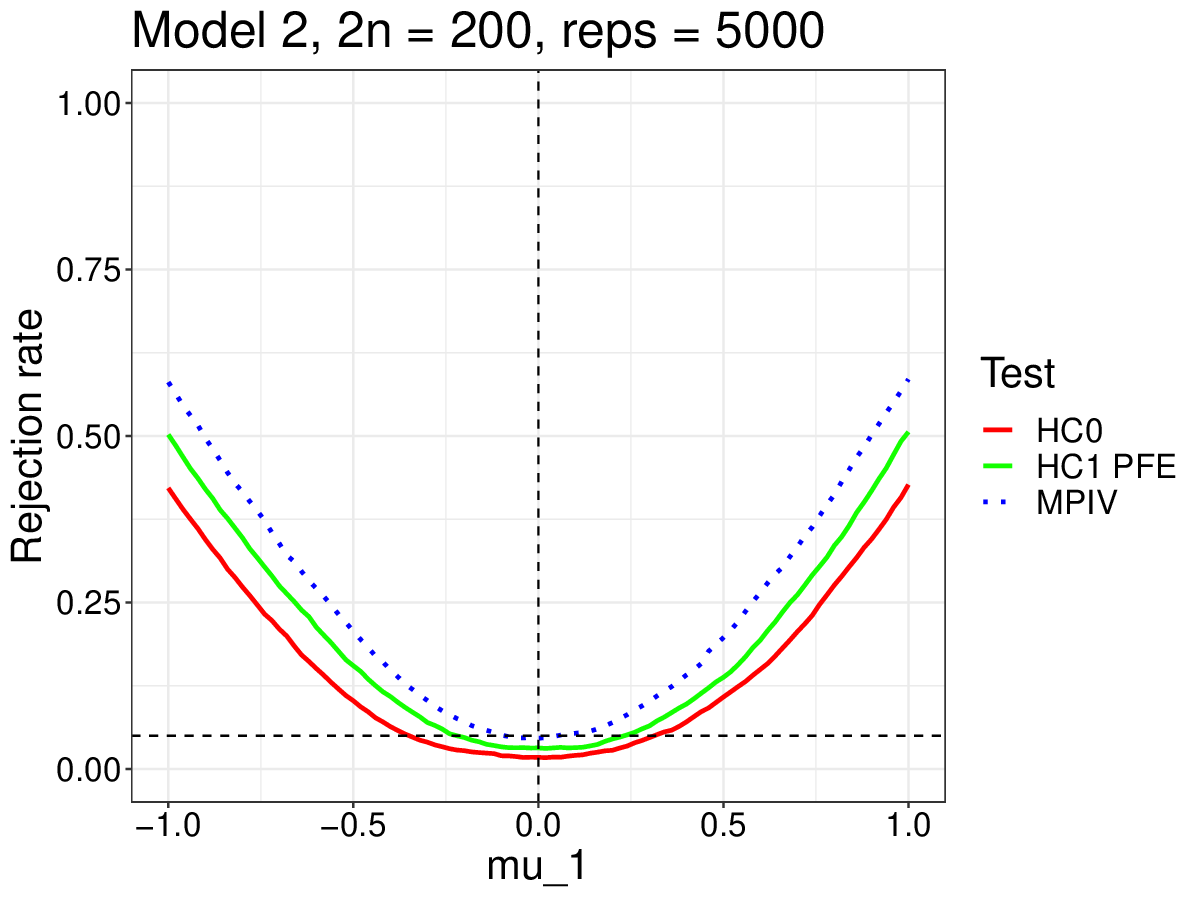}\hfill
\includegraphics[width=.5\textwidth]{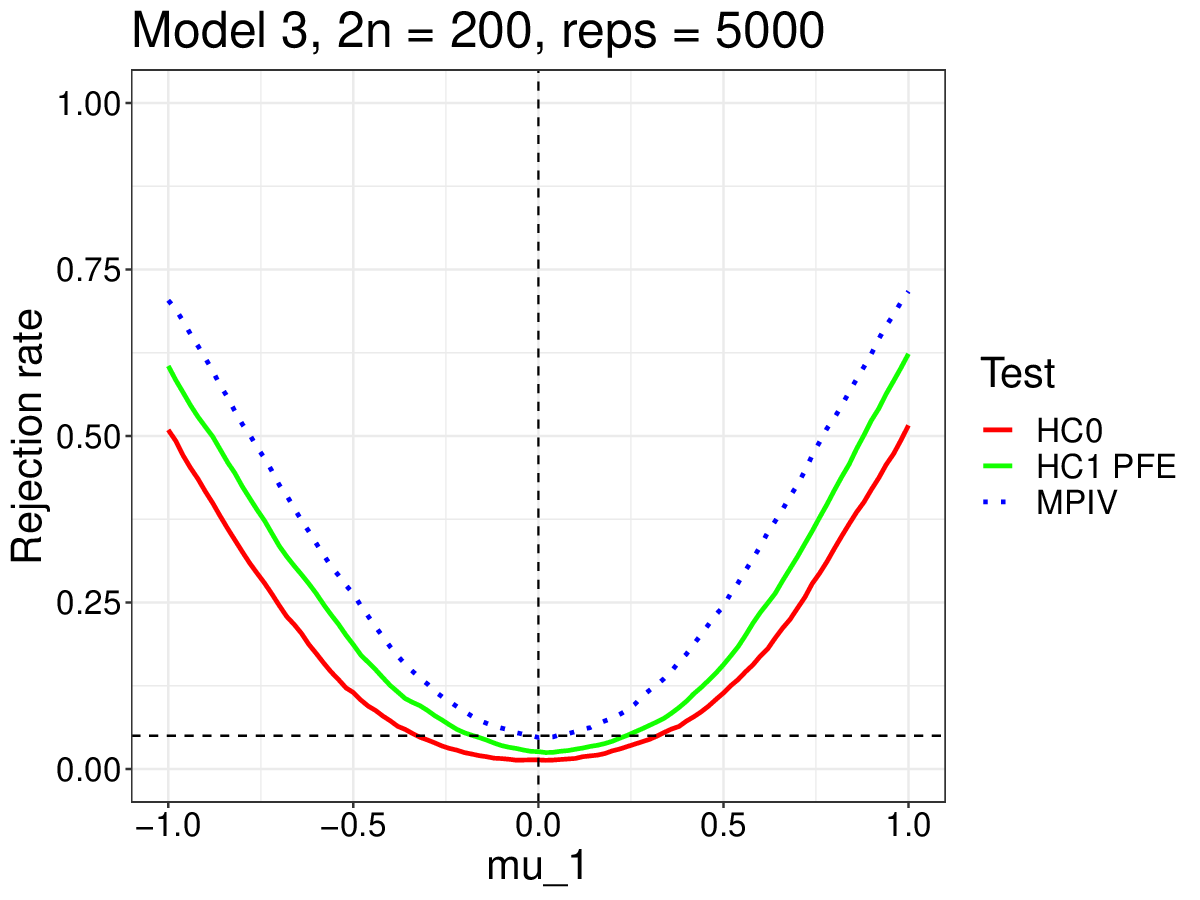}\hfill

\caption{Power curves for Models 1, 2, 3}
\label{figure:power-curve}
\end{figure}

\subsection{With Covariate Adjustment}
In this section, we examine the finite-sample behavior of the estimation and inference procedures introduced in Section \ref{sec:adjustment}. Following the simulation designs in \cite{bai2023covariate}, the potential outcomes and take-up decisions are generated according to the equations:
\begin{align*}
    Y_i(d) &= \mu_d +  m_d(X_i, W_i) + \sigma_d(X_i, W_i) \epsilon_{d, i} \\
    D_i(0) &= I \{ 0.2 X_i + 0.2 W_i X_i > \epsilon_{3, i}\} \\
    D_i(1) &= 
    \begin{cases}
        I \{0.75 + 0.2 X_i + 0.2 W_i X_i >  \epsilon_{4, i}\} & \text{if } D_i(0) = 0 \\
        1 & \text{otherwise}~,
    \end{cases}
\end{align*}
where $\mu_d$, $m_d(X_i, W_i)$, $\sigma_d(X_i, W_i)$, and $\epsilon_{d, i}$ are specified in each model below. We consider the following model specifications: 

\begin{itemize}
    \item []
    \begin{itemize}
        \item []
        \begin{itemize}

            \item [\textbf{\textbf{Model 1}:}] $(X_i, W_i)' = (\Phi(V_{1i}), \Phi(V_{2i}))'$, where $\Phi(\cdot)$ is the standard normal cumulative distribution function and
            \begin{align*}
                V_i = (V_{1i}, V_{2i})' \sim 
                N\left( 
                \begin{pmatrix}
                    0 \\ 0
                \end{pmatrix},
                \begin{pmatrix}
                    1 & \rho \\
                    \rho & 1
                \end{pmatrix}
                \right)~;
            \end{align*}
            $m_0(X_i, W_i) = m_1(X_i, W_i) = \gamma \left( W_i - \frac{1}{2} \right)$; $\epsilon_{0, i}, \epsilon_{1,i} \sim N(0, 1)$; $\epsilon_{3, i}, \epsilon_{4, i} \sim U[0,1]$; $\sigma_0(X_i, W_i) = \sigma_1(X_i, W_i) = 1$; $\mu_0 = 0; \mu_1 \in \{0, 0.5\}$; $\gamma = 4$; $\rho = 0.2$.

            \item [\textbf{\textbf{Model 2}:}] As in Model 1, but $m_0(X_i, W_i) = m_1(X_i, W_i) = \exp\left(\gamma \left(W_i - \frac{1}{2} \right)\right)$.
            
            \item [\textbf{\textbf{Model 3}:}]

            As in Model 1, but $(X_i, W_i)' = (V_{1i}, V_{1i} V_{2i})'$, and $m_0(X_i, W_i) = m_1(X_i, W_i) = \gamma_1(W_i - \rho) + \gamma_2\left( \Phi(W_i) - \frac{1}{2} \right) + \gamma_3 (X_i^2 - 1)$; $(\gamma_1, \gamma_2, \gamma_3) = (2, 1, 2)$.

            \item [\textbf{\textbf{Model 4}:}] As in Model 3, but $m_1(X_i, W_i) = m_0(X_i, W_i) + \left( \Phi(X_i) - \frac{1}{2} \right)$.
        \end{itemize}
    \end{itemize}
\end{itemize}

As in Section \ref{sec:sims_no}, for each model, let $\Delta_0$ denote the value of the LATE when $\mu_1 = 0$.\footnote{For each model, $\Delta_0$ is computed numerically with a large sample. The values are -0.0007846080, -0.0005474909, -0.0013187170,  0.0224019752 for Models 1-4.} We construct pairs by sorting units according to $X_i$ and matching adjacent units. Table \ref{table:rej-prob-models1234-XW} reports the rejection probabilities from 5,000 Monte Carlo replications for testing the null hypothesis \eqref{eq:H0} against the alternative implied by setting $\mu_1 = 1/2$, using $t$-tests constructed using three different linear covariate-adjusted estimators: the unadjusted estimator $\hat{\Delta}_n$ (denoted by \textbf{unadj}), the estimator obtained from a two-stage least squares regression of $Y_i$ on a constant, $D_i$, and covariates $W_i$, but \emph{without} pair fixed effects (denoted by \textbf{naive}), and the optimal linear estimator described in \eqref{eq:pfe-iv} with $\zeta_i = W_i$ (denoted by \textbf{pfe}). Note that each of these estimators corresponds to a special case of our regression adjusted estimator $\hat{\Delta}_n^{\rm adj}$ for three different working models.

As expected given our theoretical results in Section \ref{sec:gen_adjust}, all three tests maintain exact size under the null hypothesis. Power under the alternative hypothesis increases for all model specifications and sample sizes as we move from the unadjusted estimator to the naive linear adjusted estimator and from the naive estimator to the optimal linear adjusted estimator. Table \ref{table:mse-models1234-XW} reports the bias and root mean-squared error (RMSE) for the same set of models and sample sizes. Here, we find that RMSE decreases as we move from the unadjusted estimator to the optimal linear estimator, with no qualitative differences in bias.

\begin{table}[ht!]
    \centering
    \begin{tabular}{ c c  c c c c c c    }
         \toprule
         \multicolumn{1}{c}{} & 
         \multicolumn{1}{c}{} &
         \multicolumn{3}{c}{$H_0: \Delta = \Delta_0$} & \multicolumn{3}{c}{$H_1: \Delta = \Delta_0 + \frac{1}{2}$}  \\
         \cmidrule(lr){3-5} \cmidrule(lr){6-8} 
         Model & Sample Size & \textbf{unadj} & \textbf{naive} & \textbf{pfe} & \textbf{unadj} & \textbf{naive} & \textbf{pfe}    \\
         \midrule
        \multirow{4}*{1} & 200 & 4.98 & 5.54 & 5.68 & 42.22 & 75.36 & 75.04\\
         & 800 & 5.00 & 5.32 & 5.26 & 93.08 & 99.90 & 99.90\\
         & 1600 & 5.34 & 4.84 & 4.88 & 99.72 & 100.00 & 100.00\\
         & 3200 & 4.12 & 4.76 & 4.76 & 100.00 & 100.00 & 100.00\\
        \addlinespace
        \multirow{4}*{2} & 200 & 4.90 & 5.58 & 5.82 & 24.70 & 52.50 & 52.48\\
         & 800 & 5.60 & 5.44 & 5.48 & 69.08 & 98.02 & 97.90\\
         & 1600 & 5.26 & 4.66 & 4.62 & 94.16 & 100.00 & 100.00\\
         & 3200 & 4.40 & 5.12 & 5.12 & 99.90 & 100.00 & 100.00\\
        \addlinespace
        \multirow{4}*{3} & 200 & 5.08 & 5.36 & 5.00 & 15.06 & 36.30 & 46.98\\
         & 800 & 5.14 & 5.30 & 5.30 & 42.74 & 90.12 & 97.96\\
         & 1600 & 4.48 & 4.84 & 4.80 & 71.26 & 99.46 & 99.98\\
         & 3200 & 5.38 & 4.38 & 4.60 & 94.32 & 100.00 & 100.00\\
        \addlinespace
        \multirow{4}*{4} & 200 & 5.00 & 5.50 & 5.00 & 14.70 & 36.02 & 46.80\\
         & 800 & 5.24 & 5.28 & 5.30 & 41.56 & 90.08 & 97.96\\
         & 1600 & 4.72 & 4.64 & 4.66 & 69.50 & 99.46 & 99.98\\
         & 3200 & 5.42 & 4.32 & 4.42 & 93.54 & 100.00 & 100.00\\
        \bottomrule
        \end{tabular}
    \caption{Rejection rates for Models 1, 2, 3, 4 with additional covariates} 
    \label{table:rej-prob-models1234-XW}
\end{table}

\begin{table}[ht!]
    \centering
    \begin{tabular}{ c c  c c c c c c    }
         \toprule
         \multicolumn{1}{c}{} & 
         \multicolumn{1}{c}{} &
         \multicolumn{3}{c}{Bias} & \multicolumn{3}{c}{$\sqrt{\text{MSE}}$}  \\
         \cmidrule(lr){3-5} \cmidrule(lr){6-8} 
         Model & Sample Size & \textbf{unadj} & \textbf{naive} & \textbf{pfe} & \textbf{unadj} & \textbf{naive} & \textbf{pfe}     \\
         \midrule
        \multirow{4}*{1} & 200 & -0.00373 & -0.00025 & -0.00066 & 0.28605 & 0.19232 & 0.19288\\
         & 800 & -0.00493 & -0.00168 & -0.00175 & 0.14479 & 0.09594 & 0.09599\\
         & 1600 & 0.00001 & 0.00049 & 0.00047 & 0.10169 & 0.06563 & 0.06568\\
         & 3200 & -0.00037 & -0.00041 & -0.00042 & 0.06927 & 0.04673 & 0.04673\\
        \addlinespace
        \multirow{4}*{2} & 200 & -0.00822 & -0.00299 & -0.00345 & 0.39866 & 0.25007 & 0.25065\\
         & 800 & -0.00744 & -0.00261 & -0.00266 & 0.20069 & 0.12358 & 0.12368\\
         & 1600 & -0.00133 & -0.00068 & -0.00070 & 0.14284 & 0.08597 & 0.08602\\
         & 3200 & -0.00104 & -0.00109 & -0.00111 & 0.09663 & 0.06114 & 0.06113\\
        \addlinespace
        \multirow{4}*{3} & 200 & 0.00010 & -0.00676 & -0.00507 & 0.58324 & 0.32797 & 0.27059\\
         & 800 & -0.00461 & 0.00126 & -0.00015 & 0.28715 & 0.15316 & 0.12362\\
         & 1600 & -0.00004 & 0.00143 & 0.00115 & 0.19616 & 0.10390 & 0.08345\\
         & 3200 & -0.00195 & 0.00066 & -0.00001 & 0.14086 & 0.07203 & 0.05864\\
        \addlinespace
        \multirow{4}*{4} & 200 & 0.00191 & -0.00502 & -0.00336 & 0.59398 & 0.33014 & 0.27195\\
         & 800 & -0.00268 & 0.00334 & 0.00188 & 0.29275 & 0.15447 & 0.12453\\
         & 1600 & 0.00168 & 0.00319 & 0.00290 & 0.20024 & 0.10477 & 0.08418\\
         & 3200 & -0.00035 & 0.00231 & 0.00163 & 0.14366 & 0.07268 & 0.05914\\
        \bottomrule
        \end{tabular}
    \caption{Bias and RMSE for Models 1, 2, 3, 4 with additional covariates} 
    \label{table:mse-models1234-XW}
\end{table}










\section{Empirical Application}\label{sec:empirical-application}

In this section, we illustrate our findings by revisiting the empirical application in \cite{groh2016macroinsurance}. \cite{groh2016macroinsurance} designed a matched-pair experiment in Egypt to study the effect on microenterprises of acquiring insurance against macroeconomic shocks.\footnote{To most closely align the dataset with our theoretical results, we made the following modifications to the dataset: (1) for each outcome variable, we drop pairs if at least one of the individuals in that pair has a missing outcome variable, (2) we drop pairs if at least one of the individuals in that pair is missing treatment assignment (the eligibility of purchasing the insurance), treatment status (whether a company actually purchased the insurance), or any baseline covariates, (3) we keep only pairs with exactly two individuals (there were 39 pairs with only one individual and one ``block" with 16 individuals), 
(4) if necessary, we drop one pair from the end of the resulting dataset to ensure that the sample size is divisible by 4. (5) to construct the pairs of pairs when computing $\hat{\nu}_{n}$ and $\hat{\nu}_{n,adj}$, we use the \texttt{R} package \texttt{nbpMatching} to match pairs of pairs such that the conditions in Theorem 4.3 of \cite{bai2022mp} are satisfied. Modifications (1)-(5) result in an average sample size reduction of 96 observations (3.29\% of total sample size) across outcomes.} The eligibility to purchase macroinsurance was offered to companies in the treatment group.
The take-up rate of purchasing in the treatment group was 37\%: among 1481 companies in the treatment group, 548 of them purchased insurance. We also note among 1480 companies in the control group, 5 of them purchased insurance as well.

Table \ref{table:empirical-application} reports the estimated LATEs for a collection of outcomes, using both the unadjusted Wald estimator $\hat{\Delta}_n$ as well as the linearly adjusted estimator $\hat{\Delta}_n^{\rm adj}$ which uses the same covariates as those considered in the analysis in \cite{groh2016macroinsurance}.\footnote{We note however that we exclude the female dummy and branchid dummies, which were used in the original regression specifications in \cite{groh2016macroinsurance}, since these are  perfectly collinear with our pair fixed effects.} For the unadjusted estimator $\hat{\Delta}_n$ we report the standard errors obtained from the regression-based variance estimators $\hat{\omega}_n^2$ and $\hat{\omega}_{n, \rm pfe, \rm HC1}^2$ as well as the standard errors obtained from our consistent variance estimators $\hat{\nu}^2_n$. For the adjusted estimator $\hat{\Delta}_n^{\rm adj}$ we report the standard errors obtained from our consistent variance estimator $\hat{\nu}^{2}_{n, \rm adj}$. Our findings are consistent with the theoretical results presented in Sections \ref{sec:main} and \ref{sec:adjustment}: for the unadjusted estimates, standard errors constructed from $\hat{\nu}^2_n$ are smaller than those constructed from $\hat{\omega}^2_n$ and comparable to those constructed from $\hat{\omega}^2_{n, \rm pfe, HC1}$. Given Theorem \ref{thm:hc1-fe}, this suggests that there is limited heterogeneity in $E[Y_i^{\ast}(1) - Y_i^{\ast}(0)|X_i]$ in this application. For the adjusted estimates, we find that the standard errors constructed from $\hat{\nu}^2_{n, \rm adj}$ are smaller than those constructed from $\hat{\nu}^2_n$ across all outcomes. However, point estimates for profits and the aggregate index change in such a way that these are no longer significant at the $10\%$ level.

\begin{table}[ht!]
\caption{Summary of Estimates Obtained from Empirical Application \cite{groh2016macroinsurance}}
\begin{adjustbox}{max width=\linewidth,center}
\begin{tabular}[t]{llllllllll}
\toprule
& & High& &High &Number &Any &Owner's &Monthly & Aggregate\\
& Profits & Profit & Revenue & Revenue & Employees & Worker & Hours & Consumption & Index\\
\midrule
$\hat \Delta_n$ & -241.602 & -0.029 & -2561.536 & -0.063 & -0.085 & 0.009 & -1.324 & -14.197 & -0.104\\
\addlinespace
SE from $\hat \omega^2_n$ & (149.696) & (0.022) & (937.455)$^{\ast\ast\ast}$ & (0.021)$^{\ast\ast\ast}$ & (0.144) & (0.050) & (2.506) & (91.237) & (0.069)\\
\addlinespace
SE from $\hat \omega^2_{n,\rm pfe, HC1}$ & (131.073)$^{\ast}$ & (0.021) & (861.356)$^{\ast\ast\ast}$ & (0.020)$^{\ast\ast\ast}$ & (0.137) & (0.048) & (2.170) & (78.825) & (0.061)$^{\ast}$\\
\addlinespace
SE from $\hat \nu^2_n$ & (130.946)$^{\ast}$ & (0.021) & (841.562)$^{\ast\ast\ast}$ & (0.020)$^{\ast\ast\ast}$ & (0.138) & (0.049) & (2.179) & (79.273) & (0.061)$^{\ast}$\\
\addlinespace[1em]
$\hat \Delta_n^{\rm adj}$ & -151.345 & -0.020 & -1802.829 & -0.052 & -0.023 & 0.024 & -0.401 & 6.541 & -0.069\\
\addlinespace
SE from $\hat \nu^2_{n, \rm adj}$ & (120.969) & (0.020) & (751.651)$^{\ast\ast}$ & (0.018)$^{\ast\ast\ast}$ & (0.133) & (0.047) & (2.116) & (74.388) & (0.057)\\
\addlinespace[1em]
Sample Size & 2804 & 2804 & 2800 & 2800 & 2824 & 2824 & 2796 & 2880 & 2880\\
\bottomrule
\end{tabular}
\end{adjustbox}
\begin{tablenotes}
\small \item Notes: $^\ast$: significant at 10\% level. $^{\ast\ast}$: significant at 5\% level. $^{\ast\ast\ast}$: significant at 1\% level. For each outcome listed in Table 7 of \cite{groh2016macroinsurance}, we report (a) the Wald estimates $\hat \Delta_n$ in \eqref{eq:hatDelta}, (b) the robust standard error obtained from $\hat \omega^2_n$ in \eqref{eq:robust}, (c) the robust standard error obtained from $\hat \omega^2_{n, \rm pfe, HC1}$ in \eqref{eq:pfe-HC1}, (d) the MPIV standard error obtained from $\hat \nu^2_n$ in \eqref{eq:mpiv-varest}, (e) the covariate-adjusted estimates $\hat \Delta_n^{\rm adj}$ with pair fixed effects based on \eqref{eq:pfe-Y} and \eqref{eq:pfe-D}, (f) the standard error obtained from $\hat \nu^2_{n, \rm adj}$ in \eqref{eq:adj-varest}, and (g) the sample sizes for the regression of each outcome variable. 
\end{tablenotes}
\label{table:empirical-application}
\end{table}

\section{Recommendations for Empirical Practice} \label{sec:recs}
Based on our theoretical results as well as the simulation study above, we conclude with some recommendations for practitioners when conducting inference about the local average treatment effect in matched-pairs experiments. Our findings are that the standard Wald estimator is generally consistent and asymptotically normal under matched-pair designs, but its limiting variance is smaller than what would be obtained under i.i.d.\ assignment.  It follows that inferences using standard heteroskedasticty-robust estimators of the variance will typically be conservative. We therefore recommend that practitioners use our consistent variance estimator $\hat{\nu}^2_n$ instead. When considering covariate adjustment, our findings are that the two-stage least squares estimator with pair fixed effects leads to an estimator that is optimal in the sense of having smallest limiting variance among all linearly-adjusted estimators. An important caveat, however, is that the usual heteroskedasticty-robust estimator of the variance is not consistent for its variance. As a result, we recommend that practitioners use our consistent variance estimator $\hat{\nu}^2_{n, \rm adj}$ instead.

\newpage
\bibliography{bibliography}

\clearpage
\appendix

\section{Proofs of Main Results}
\subsection{Proof of Theorem \ref{theorem:main}}
\begin{proof}
\begin{align*}
\sqrt{n} \left( \hat \Delta_n - \Delta(Q) \right) &= \sqrt{n} \left( \frac{\hat{\psi}_n(1) - \hat{\psi}_n(0)}{\hat{\phi}_n(1) - \hat{\phi}_n(0)} - \Delta(Q) \right) \\
&= \frac{\sqrt{n} \left( \hat{\psi}_n(1) - \hat{\psi}_n(0) - \Delta(Q) \left( \hat{\phi}_n(1) - \hat{\phi}_n(0)  \right) \right)}{\hat{\phi}_n(1) - \hat{\phi}_n(0) }~.
\end{align*}
Following identical arguments to the proof of Lemma S.1.5 in \cite{bai2022mp}, we have
\begin{equation} \label{eq:lln}
\hat{\phi}_n(a) = \frac{1}{n}\sum_{1 \le i \le 2n:A_i = a}D_i \xrightarrow{P} E[D_i(a)]~,
\end{equation}
from which it follows by the continuous mapping theorem and Assumption \ref{ass:Q}(d) that
\begin{equation}\label{eq:phi-hat-cp-pc}
\hat \phi_n(1) - \hat \phi_n(0) \cp P \{C_i = 1\}~.
\end{equation}
Therefore, to understand the limiting distribution of $\hat \Delta_n$, it suffices to show that
\begin{equation}\label{eq:psi-tilde-clt}
\sqrt{n} \left( \hat{\psi}_n(1) - \hat{\psi}_n(0) - \Delta(Q) \left( \hat{\phi}_n(1) - \hat{\phi}_n(0)  \right) \right) = \sqrt n(\tilde \psi_n(1) - \tilde \psi_n(0)) \cd N(0, \nu^2 P \{C_i = 1\}^2)~,
\end{equation}
where
\[ \tilde \psi_n(a) = \frac{1}{n} \sum_{1\leq i\leq 2n: A_i = a} Y_i^\ast(a)~, \]
and we recall that 
\begin{equation}\label{eq:Ystar-potential}
Y^\ast_i(a) = \tilde Y_i(a) - \Delta(Q)D_i(a)~.
\end{equation}

We may write 
\begin{align}\label{eq:basic_decomp}
\sqrt{n} \left( \tilde \psi_n(1) - \tilde \psi_n(0) \right) = A_n - B_n + C_n - D_n~,
\end{align}
where
\begin{align*}
A_n &= \frac{1}{\sqrt{n}} \sum_{1\leq i\leq 2n} \left( Y_i^\ast(1) A_i - E[Y_i^\ast(1) A_i | X^{(n)}, A^{(n)}] \right)~,\\
B_n &= \frac{1}{\sqrt{n}} \sum_{1\leq i\leq 2n} \left( Y_i^\ast(0) (1 - A_i) - E[Y_i^\ast(0) (1 - A_i) | X^{(n)}, A^{(n)}] \right)~,\\
C_n &= \frac{1}{\sqrt{n}} \sum_{1\leq i\leq 2n} \left( E[Y_i^\ast(1) A_i | X^{(n)}, A^{(n)}] - A_i E[Y_i^\ast(1)] \right)~,\\
D_n &= \frac{1}{\sqrt{n}} \sum_{1\leq i\leq 2n} \left( E[Y_i^\ast(0) (1 - A_i) | X^{(n)}, A^{(n)}] - (1 - A_i) E[Y_i^\ast(0)] \right)~.
\end{align*}
Note that the decomposition \eqref{eq:basic_decomp} holds because
\begin{equation}\label{eq:Ystar1-Ystar0}
\begin{split}
E[Y_i^\ast(1)] - E[Y_i^\ast(0)] 
&= E[\tilde Y_i(1) - \Delta(Q)D_i(1)] - E[\tilde Y_i(0) - \Delta(Q)D_i(0)] \\
&= E[\tilde Y_i(1) - \tilde Y_i(0)] - \Delta(Q) E[D_i(1) - D_i(0)] \\
&= E[(Y_i(1) - Y_i(0)) (D_i(1) - D_i(0))] - \Delta(Q) E[D_i(1) - D_i(0)] \\
&= 0~,
\end{split}
\end{equation}
where the first equation follows by (\ref{eq:Ystar-potential}), the second equation follows by inspection, the third equation follows by (\ref{eq:TildeY}), and the last equation follows by the definition of $\Delta(Q)$. 

Given the decomposition \eqref{eq:basic_decomp}, the proof of (\ref{eq:psi-tilde-clt}) follows identically to the proof of Lemma S.1.4 in \cite{bai2022mp} using the transformed outcomes $Y_i^*(a)$ in place of the original potential outcomes. For completeness, we sketch the proof here and refer the reader to \cite{bai2022mp} for details. It can be shown using Lindeberg's central limit theorem applied conditionally on $X^{(n)}, A^{(n)}$ that 
\[\sup_{t \in \mathbb{R}} \left | P\{A_n - B_n \le t |X^{(n)}, A^{(n)}\} - \Phi\left(t/\sqrt{E[\var[Y_i^*(0)|X_i]]+E[\var[Y_i^*(1)|X_i]]}\right)\right| \xrightarrow{P} 0~.\] Next, it can be shown that, unconditionally,
\[C_n - D_n \xrightarrow{d} N\left(0, \frac{1}{2}E\left[((E[Y^*_i(1)|X_i] - E[Y^*_i(1)]) - (E[Y^*_i(0)|X_i] - E[Y^*_i(0)]))^2\right]\right)~.\]
Further note $C_n$ and $D_n$ are functions of $X^{(n)}$ and $A^{(n)}$ only. Therefore, combining the above two expressions using, for instance, Lemma S.1.2 in \cite{bai2022mp}, we obtain \eqref{eq:psi-tilde-clt}. By Slutsky's theorem, the desired conclusion follows by (\ref{eq:phi-hat-cp-pc}) and (\ref{eq:psi-tilde-clt}) under Assumption \ref{ass:Q}(e).
\end{proof}

\subsection{Proof of Theorem \ref{theorem:varianceestimation}}
\begin{proof}
First note that as a consequence of \eqref{eq:phi-hat-cp-pc} and the continuous mapping theorem, 
\[\left(\hat{\phi}_n(1) - \hat{\phi}_n(0)\right)^2 \cp P \{C_i = 1\}^2~.\]
It thus suffices to show that the numerator converges to the desired quantity.
Denote the (infeasible) adjusted observed outcome as
\begin{equation}\label{eq:Ystar}
Y^\ast_i = Y_i - \Delta(Q) D_i~.
\end{equation}
Note that the observed adjusted outcome can alternatively be rewritten as
\begin{equation}\label{eq:obs-adj-Y}
Y_i^\ast = Y_i^\ast(1) A_i + Y_i^\ast(0) (1 - A_i)~.
\end{equation}
Consider the following infeasible version of the numerator, given by 
\[\tilde{\tau}^2_n - \frac{1}{2}(\tilde{\lambda}^2_n + \tilde{\Gamma}^2_n)~,\]
where 
\[\tilde{\tau}^2_n = \frac{1}{n}\sum_{1 \le j \le n}(Y^\ast_{\pi(2j)} - Y^\ast_{\pi(2j-1)})^2~,\]
\[\tilde{\lambda}^2_n = \frac{2}{n}\sum_{1 \le j \le \lfloor \frac{n}{2} \rfloor}\left(Y^\ast_{\pi(4j-3)} - Y^\ast_{\pi(4j-2)}\right)\left(Y^\ast_{\pi(4j-1)} - Y^\ast_{\pi(4j)}\right)\left(A_{\pi(4j-3)} - A_{\pi(4j-2)}\right)\left(A_{\pi(4j-1)} - A_{\pi(4j)}\right)~,\]
\[\tilde{\Gamma}_n = \frac{1}{n}\sum_{1 \le i \le 2n: A_i = 1}Y_i^\ast - \frac{1}{n}\sum_{1 \le i \le 2n: A_i = 0}Y_i^\ast~.\]
It follows immediately from Assumption \ref{ass:Q} that Assumption 2.1(b)--(c) in \cite{bai2022mp} are satisfied for the transformed outcomes $Y_i^\ast(a)$, and thus it follows from Lemmas S.1.5, S.1.6, and S.1.7 in \cite{bai2022mp} that this infeasible numerator converges to the desired quantity. It thus remains to show that 
\begin{align}
\label{eq:tau} \hat{\tau}^2_n & = \tilde{\tau}^2_n + o_P(1) \\
\label{eq:lambda-P} \hat{\lambda}^2_n & = \tilde{\lambda}^2_n + o_P(1) \\
\label{eq:gamma-P} \hat{\Gamma}_n & = \tilde{\Gamma}_n + o_P(1)~.
\end{align}
We begin with \eqref{eq:tau}. To see this, note that 
\begin{multline*}
\frac{1}{n}\sum_{1 \le j \le n}\left\{(Y^\ast_{\pi(2j)} - Y^\ast_{\pi(2j-1)})^2 - (\hat{Y}_{\pi(2j)} - \hat{Y}_{\pi(2j-1)})^2\right\} \\
= \frac{1}{n}\sum_{1 \le j \le n}\left\{(Y^\ast_{\pi(2j)} - Y^\ast_{\pi(2j-1)})^2 - (\hat{Y}_{\pi(2j)} - \hat{Y}_{\pi(2j-1)})^2\right\}A_{\pi(2j)} \\
+ \frac{1}{n}\sum_{1 \le j \le n}\left\{(Y^\ast_{\pi(2j)} - Y^\ast_{\pi(2j-1)})^2 - (\hat{Y}_{\pi(2j)} - \hat{Y}_{\pi(2j-1)})^2\right\}A_{\pi(2j-1)}~.
\end{multline*}
It thus suffices to show that each component on the RHS converges in probability to zero. We only show the first since the second follows symmetrically. From the definitions of $Y^\ast_i$ and $\hat{Y}_i$ we obtain that 
\begin{align*}
& \frac{1}{n}\sum_{1 \le j \le n}\left\{(Y^\ast_{\pi(2j)} - Y^\ast_{\pi(2j-1)})^2 - (\hat{Y}_{\pi(2j)} - \hat{Y}_{\pi(2j-1)})^2\right\}A_{\pi(2j)} \\
& = \frac{1}{n}\sum_{1 \le j \le n}\Big\{(\tilde Y_{\pi(2j)}(1) - \tilde Y_{\pi(2j-1)}(0) - \Delta(Q)(D_{\pi(2j)}(1) - D_{\pi(2j-1)}(0)))^2 \\
& \hspace{3em} - (\tilde Y_{\pi(2j)}(1) - \tilde Y_{\pi(2j-1)}(0) - \hat{\Delta}_n(D_{\pi(2j)}(1) - D_{\pi(2j-1)}(0)))^2 \Big\}A_{\pi(2j)} \\
& = -2(\Delta(Q) - \hat{\Delta}_n)\frac{1}{n}\sum_{1 \le j \le n}\left\{(\tilde Y_{\pi(2j)}(1) - \tilde Y_{\pi(2j-1)}(0))((D_{\pi(2j)}(1) - D_{\pi(2j-1)}(0))\right\}A_{\pi(2j)} \\
& \hspace{3em} + (\Delta(Q)^2 - \hat{\Delta}_n^2)\frac{1}{n}\sum_{1 \le j \le n}\left\{(D_{\pi(2j)}(1) - D_{\pi(2j-1)}(0))^2\right\}A_{\pi(2j)}
\end{align*}
Next, note that by the triangle inequality, Assumption \ref{ass:Q}(b) and the weak law of large numbers, 
\begin{multline*}
\left|\frac{1}{n}\sum_{1 \le j \le n}\left\{(\tilde Y_{\pi(2j)}(1) - \tilde Y_{\pi(2j-1)}(0))((D_{\pi(2j)}(1) - D_{\pi(2j-1)}(0))\right\}A_{\pi(2j)}\right| \\
\le \frac{1}{n}\sum_{1 \le i \le 2n}|\tilde Y_i(1)| + \frac{1}{n}\sum_{1 \le i \le 2n}|\tilde Y_i(0)| = O_P(1)~,  
\end{multline*}
and since $D$ and $A$ are binary,
\[\left|\frac{1}{n}\sum_{1 \le j \le n}\left\{(D_{\pi(2j)}(1) - D_{\pi(2j-1)}(0))^2\right\}A_{\pi(2j)}\right| \le 1~,\]
and hence the result follows from the fact that $\hat{\Delta}_n \cp \Delta(Q)$ by Theorem \ref{theorem:main}.
To show \eqref{eq:lambda-P}, note that
\begin{align*}
& \hat \lambda_n^2 - \tilde \lambda_n^2 \\
& = (\hat \Delta_n - \Delta(Q))^2 \frac{2}{n} \sum_{1\leq j\leq \lfloor \frac{n}{2} \rfloor} ( D_{\pi(4j-3)} - D_{\pi(4j-2)} ) ( D_{\pi(4j-1)} - D_{\pi(4j)} ) \\
& \hspace{5em} \times (A_{\pi(4j-3)} - A_{\pi(4j-2)}) (A_{\pi(4j-1)} - A_{\pi(4j)}) \\
&\hspace{0.5em} - (\hat \Delta_n - \Delta(Q)) \frac{2}{n}\sum_{1 \le j \le \lfloor \frac{n}{2} \rfloor} (D_{\pi(4j-3)} - D_{\pi(4j-2)}) (Y^\ast_{\pi(4j-1)} - Y^\ast_{\pi(4j)}) \\
& \hspace{5em} \times (A_{\pi(4j-3)} - A_{\pi(4j-2)}) (A_{\pi(4j-1)} - A_{\pi(4j)}) \\
& \hspace{0.5em} - (\hat \Delta_n - \Delta(Q)) \frac{2}{n}\sum_{1 \le j \le \lfloor \frac{n}{2} \rfloor} (D_{\pi(4j-1)} - D_{\pi(4j)}) (Y^\ast_{\pi(4j-3)} - Y^\ast_{\pi(4j-2)}) \\
& \hspace{5em} \times (A_{\pi(4j-3)} - A_{\pi(4j-2)}) (A_{\pi(4j-1)} - A_{\pi(4j)})~.
\end{align*}
Note $\hat \Delta_n \cp \Delta(Q)$ because of Theorem \ref{theorem:main}. On the other hand,
\begin{align*}
& \Big | \frac{2}{n}\sum_{1 \le j \le \lfloor \frac{n}{2} \rfloor} (D_{\pi(4j-3)} - D_{\pi(4j-2)}) (Y^\ast_{\pi(4j-1)} - Y^\ast_{\pi(4j)}) (A_{\pi(4j-3)} - A_{\pi(4j-2)}) (A_{\pi(4j-1)} - A_{\pi(4j)}) \Big | \\
& \leq \frac{2}{n} \sum_{1 \le j \le \lfloor \frac{n}{2} \rfloor} (|Y^\ast_{\pi(4j-1)}| + |Y^\ast_{\pi(4j)}|) \\
& \leq \frac{2}{n} \sum_{1 \leq i \leq 2n} |Y_i^\ast(1)| + \frac{2}{n} \sum_{1 \leq i \leq 2n} |Y_i^\ast(0)|~,
\end{align*}
where the first inequality follows from the triangle inequality and the fact that $D$ and $A$ are binary, and the last inequality follows trivially. And since $D$ and $A$ are binary,
\[ \left| \frac{2}{n} \sum_{1\leq j\leq \lfloor \frac{n}{2} \rfloor} ( D_{\pi(4j-3)} - D_{\pi(4j-2)} ) ( D_{\pi(4j-1)} - D_{\pi(4j)} ) (A_{\pi(4j-3)} - A_{\pi(4j-2)}) (A_{\pi(4j-1)} - A_{\pi(4j)}) \right| \leq 1~. \]
\eqref{eq:lambda-P} then follows because
\[ \frac{2}{n} \sum_{1 \leq i \leq 2n} |Y_i^\ast(a)| = O_P(1) \]
for $a \in \{0, 1\}$ because of Assumption \ref{ass:Q}(b) and the weak law of large numbers. Finally, \eqref{eq:gamma-P} follows immediately from $\hat \Delta_n \cp \Delta(Q)$.
\end{proof}

\subsection{Proof of Theorem \ref{theorem:robust}}
\begin{proof}
Let $\hat U_i$ denote the $i$th residual generated by the two-stage least squares  estimator in a linear regression of $Y_i$ on a constant and $D_i$ using $A_i$ as an instrument for $D_i$. Recall the heteroskedasticity-robust variance estimator is defined as
\begin{equation} \label{eq:robust}
\hat \omega_n^2 = \left ( n \left ( \sum_{1 \leq i \leq 2n} \begin{pmatrix}
1 \\ A_i
\end{pmatrix}
\begin{pmatrix}
1 & D_i
\end{pmatrix} \right )^{-1} \sum_{1 \leq i \leq 2n} \hat U_i^2 \begin{pmatrix}
1 \\ A_i
\end{pmatrix}
\begin{pmatrix}
1 & A_i
\end{pmatrix} 
\left ( \sum_{1 \leq i \leq 2n} \begin{pmatrix}
1 \\ D_i
\end{pmatrix}
\begin{pmatrix}
1 & A_i
\end{pmatrix} \right )^{-1}
\right )_{2, 2}~,
\end{equation}
where the notation $(\cdot)_{2,2}$ denotes the $(2,2)$-element of its (matrix) argument.

Following identical arguments to those used to show \eqref{eq:lln}, it can be shown that 
\begin{align*}
\frac{1}{n} \sum_{1 \leq i \leq 2n} I\{ A_i = a\} D_i & \cp E[D_i(a)] \\
\frac{1}{n} \sum_{1 \leq i \leq 2n} I\{ A_i = a\} Y_i^r & \cp E[\tilde Y_i^r(a)] 
\end{align*}
for $r = 1, 2$. Note in addition that
\[ \hat U_i = Y_i - \frac{1}{2n} \sum_{1 \leq i \leq 2n} Y_i - \bigg ( D_i - \frac{1}{2n} \sum_{1 \leq i \leq 2n} D_i \bigg ) \hat \Delta_n~. \]
It follows from direct calculation that
\[ \hat \omega_n^2 = \frac{\displaystyle \frac{1}{n} \sum_{1 \leq i \leq 2n} \hat U_i^2}{\displaystyle \bigg ( \frac{2}{n} \sum_{1 \leq i \leq 2n} A_i D_i - \frac{1}{n} \sum_{1 \leq i \leq 2n} D_i \bigg )^2}~. \]
The conclusion then follows from the above derivations, the continuous mapping theorem, and additional direct calculations.
\end{proof}

\subsection{Proof of Theorem \ref{thm:hc1-fe}}
\begin{proof}
It follows from Lemma \ref{lem:FE_TSLS} applied without additional covariates that
\[ \hat \alpha_n = \hat \Delta_n~. \]
It follows from the orthogonality condition
\[ \sum_{1 \leq i \leq n} \left ( Y_i - \hat \alpha_n D_i - \sum_{1 \leq j \leq n} \hat \theta_{j, n} I \{i \in \{\pi(2j - 1), \pi(2j)\} \right ) I \{i \in \{\pi(2j - 1), \pi(2j)\} = 0 \]
that
\[ \hat \theta_{j, n} = \frac{1}{2} (\hat Y_{\pi(2j - 1)} + \hat Y_{\pi(2j)})~, \]
where $\hat Y_i$ is the transformed outcome defined in \eqref{eq:transformed}. As a result,
\begin{align*}
\hat \epsilon_{\pi(2j - 1)} & = \frac{1}{2} (\hat Y_{\pi(2j - 1)} - \hat Y_{\pi(2j)}) \\
\hat \epsilon_{\pi(2j)} & = \frac{1}{2} (\hat Y_{\pi(2j)} - \hat Y_{\pi(2j - 1)})~.
\end{align*}
Further note the residual of the projection of $D$ on the pair fixed effects is $(D_{\pi(2j - 1)} - D_{\pi(2j)}) / 2$ for $2j - 1$ and similarly for $2j$. The same goes for the residual of the projection of $A$ on the pair fixed effects. It then follows from Lemma \ref{lem:fwl} that the the heteroskedasticity-robust variance estimator for $\hat \alpha_n$ is
\begin{align*}
& \frac{\displaystyle \frac{1}{n} \frac{1}{8} \sum_{1 \leq j \leq n} (A_{\pi(2j - 1)} - A_{\pi(2j)})^2 (\hat Y_{\pi(2j - 1)} - \hat Y_{\pi(2j)})^2}{\displaystyle \left ( \frac{1}{n} \frac{1}{2} \sum_{1 \leq j \leq n} (A_{\pi(2j - 1)} - A_{\pi(2j)})(D_{\pi(2j - 1)} - D_{\pi(2j)}) \right )^2} \\
& = \frac{1}{2} \frac{\hat \tau_n^2}{(\hat \phi_n(1) - \hat \phi_n(0))^2}~.
\end{align*}
With HC1,
\[ \frac{2n}{2n - (1 + n)} \frac{1}{2} \frac{\hat \tau_n^2}{(\hat \phi_n(1) - \hat \phi_n(0))^2} = \frac{2n}{n - 1} \frac{1}{2} \frac{\hat \tau_n^2}{(\hat \phi_n(1) - \hat \phi_n(0))^2} = \frac{n}{n - 1} \frac{\hat \tau_n^2}{(\hat \phi_n(1) - \hat \phi_n(0))^2}~. \]
The conclusion then follows noting the limits in probability of all individual terms have been established in the proofs of earlier theorems.
\end{proof}

\subsection{Proof of Theorem \ref{theorem:adjustment}}
\begin{proof}
Note
\[ \sqrt n (\hat \Delta_n^{\rm adj} - \Delta(Q)) = \frac{\sqrt n(\hat \psi_n^{\rm adj}(1) - \hat \psi_n^{\rm adj}(0) - \Delta(Q)(\hat \phi_n^{\rm adj}(1) - \hat \phi_n^{\rm adj}(0)))}{\hat \phi_n^{\rm adj}(1) - \hat \phi_n^{\rm adj}(0)}~. \]
Following identical arguments to the proof of Theorem 3.1 in \cite{bai2023covariate},  under \eqref{eq:m-D-close},
\[ \hat \phi_n^{\rm adj}(1) - \hat \phi_n^{\rm adj}(0) - P \{C_i = 1\} = O_P(n^{-1/2})~, \]
and hence
\begin{equation}\label{eq:adj-phi-hat-cp-pc}
\hat \phi_n^{\rm adj}(1) - \hat \phi_n^{\rm adj}(0) \cp P \{C_i = 1\}~.
\end{equation}
Therefore, to understand the limiting distribution of $\hat \Delta_n^{\rm adj}$, it suffices to show that
\begin{equation}\label{eq:adj-clt}
\sqrt n(\hat \psi_n^{\rm adj}(1) - \hat \psi_n^{\rm adj}(0) - \Delta(Q)(\hat \phi_n^{\rm adj}(1) - \hat \phi_n^{\rm adj}(0))) = \sqrt n(\tilde \psi_n^{\rm adj}(1) - \tilde \psi_n^{\rm adj}(0)) \cd N(0, \nu_{1, \rm adj}^2 + \nu_{2, \rm adj}^2 + \nu_{3, \rm adj}^2)~,
\end{equation}
where
\begin{align*}
\tilde \psi_n^{\rm adj}(a) &= \frac{1}{2n} \sum_{1 \leq i \leq 2n} (2 I \{A_i = a\} (Y_i^\ast - \hat m^\ast_{a, \tilde Y D}(X_i, W_i)) 
+ \hat m^\ast_{a, \tilde Y D}(X_i, W_i))~, \\
\hat m^\ast_{a, \tilde Y D}(X_i, W_i) &= \hat m_{a, \tilde Y}(X_i, W_i) - \Delta(Q) \hat m_{a, D}(X_i, W_i)~.
\end{align*}
To show \eqref{eq:adj-clt} we will apply the arguments in the proof of Theorem 3.1 in \cite{bai2023covariate}. In order to do so, it suffices to decompose the left-hand side of \eqref{eq:adj-clt} in the following way. First note that \eqref{eq:m-Y-hat-close} and \eqref{eq:m-D-close} imply that
\begin{equation}\label{eq:m-star-hat-close}
\frac{1}{\sqrt{2n}} \sum_{1\leq i\leq 2n} (2A_i - 1) (\hat m^\ast_{a, \tilde Y D}(X_i, W_i) - m_{a, \tilde Y D}(X_i, W_i)) \cp 0~,
\end{equation}
so that
\begin{align*}
\tilde \psi_n^{\rm adj}(1) &= \frac{1}{2n} \sum_{1\leq i\leq 2n} (2 A_i (Y_i^\ast(1) - \hat m^\ast_{1, \tilde Y D}(X_i, W_i)) + \hat m^\ast_{1, \tilde Y D}(X_i, W_i) ) \\
&= \frac{1}{2n} \sum_{1\leq i\leq 2n} (2 A_i Y_i^\ast(1) - (2A_i - 1) \hat m^\ast_{1, \tilde Y D}(X_i, W_i) ) \\
&= \frac{1}{2n} \sum_{1\leq i\leq 2n} (2 A_i Y_i^\ast(1) - (2A_i - 1) m_{1, \tilde Y D}(X_i, W_i) ) + o_P(n^{-1/2}) \\
&= \frac{1}{2n} \sum_{1\leq i\leq 2n} (2 A_i Y_i^\ast(1) - A_i m_{1, \tilde Y D}(X_i, W_i) - (1 - A_i) m_{1, \tilde Y D}(X_i, W_i) ) + o_P(n^{-1/2})~,
\end{align*}
where the third equality follows from (\ref{eq:m-star-hat-close}). Similarly,
\begin{align*}
\tilde \psi_n^{\rm adj}(0) = \frac{1}{2n} \sum_{1\leq i\leq 2n} ( 2(1 - A_i) Y_i^\ast(0) - A_i m_{0, \tilde Y D}(X_i, W_i) - (1 - A_i) m_{0, \tilde Y D}(X_i, W_i) ) + o_P(n^{-1/2})~.
\end{align*}
Then
\begin{align*}
\tilde \psi_n^{\rm adj}(1) - \tilde \psi_n^{\rm adj}(0) &= \frac{1}{n} \sum_{1\leq i\leq 2n} A_i \phi^\ast_{1, i} - \frac{1}{n} \sum_{1\leq i\leq 2n} (1 - A_i) \phi^\ast_{0, i} + o_P(n^{-1/2})~,
\end{align*}
where
\begin{align*}
\phi^\ast_{1, i} &= Y_i^\ast(1) - \frac{1}{2} (m_{1, \tilde Y D}(X_i, W_i) + m_{0, \tilde Y D}(X_i, W_i) )~, \\
\phi^\ast_{0, i} &= Y_i^\ast(0) - \frac{1}{2} (m_{1, \tilde Y D}(X_i, W_i) + m_{0, \tilde Y D}(X_i, W_i) )~.
\end{align*}
Consider the decomposition
\begin{align*}
\sqrt{n} (\tilde \psi_n^{\rm adj}(1) - \tilde \psi_n^{\rm adj}(0)) = A^{\rm adj}_n - B^{\rm adj}_n + C^{\rm adj}_n - D^{\rm adj}_n + o_P(n^{-1/2})~,
\end{align*}
where
\begin{align*}
A^{\rm adj}_n &= \frac{1}{\sqrt{n}} \sum_{1\leq i\leq 2n} (A_i \phi^\ast_{1, i} - E[A_i \phi^\ast_{1, i} | X^{(n)}, A^{(n)}])~, \\
B^{\rm adj}_n &= \frac{1}{\sqrt{n}} \sum_{1\leq i\leq 2n} ( (1 - A_i) \phi^\ast_{1, i} - E[ (1 - A_i) \phi^\ast_{1, i} | X^{(n)}, A^{(n)}])~, \\
C^{\rm adj}_n &= \frac{1}{\sqrt{n}} \sum_{1\leq i\leq 2n} A_i (E[Y^\ast_i(1)|X_i] - E[Y^\ast_i(1)] )~, \\
D^{\rm adj}_n &= \frac{1}{\sqrt{n}} \sum_{1\leq i\leq 2n} (1 - A_i) (E[Y^\ast_i(0)|X_i] - E[Y^\ast_i(0)] )~,
\end{align*}
where we note that, as in the proof of Theorem \ref{theorem:main}, the decomposition holds because of \eqref{eq:Ystar1-Ystar0}. \eqref{eq:adj-clt} then follows by deriving the limit of $A^{\rm adj}_n - B^{\rm adj}_n$ conditional on $X^{(n)}, A^{(n)}$ and the unconditional limit of $C_n^{\rm adj} - D_n^{\rm adj}$, as in the proof of Theorem \ref{theorem:main}. Finally, by Slutsky's lemma, the conclusion of the theorem follows from \eqref{eq:adj-phi-hat-cp-pc} and \eqref{eq:adj-clt} under Assumption \ref{ass:Q}(e).
\end{proof}

\subsection{Proof of Theorem \ref{theorem:varianceestimation-adj}}
\begin{proof}
First note that (\ref{eq:adj-phi-hat-cp-pc}) and continuous mapping theorem imply 
\[\left(\hat{\phi}^{\rm adj}_n(1) - \hat{\phi}^{\rm adj}_n(0)\right)^2 \cp P \{C_i = 1\}^2~.\]
It thus suffices to show that the numerator converges to the desired quantity.

Consider the following infeasible version of the numerator, given by 
\[\tilde{\tau}^2_{n, \rm adj} - \frac{1}{2}(\tilde{\lambda}^2_{n, \rm adj} + \tilde{\Gamma}^2_{n, \rm adj})~,\]
where
\begin{align*}
\tilde{\tau}^2_{n, \rm adj} & = \frac{1}{n}\sum_{1 \le j \le n}(\hat Y^\ast_{\pi(2j), \rm adj} - \hat Y^\ast_{\pi(2j-1), \rm adj})^2~, \\
\tilde{\lambda}^2_{n, \rm adj} & = \frac{2}{n}\sum_{1 \le j \le \lfloor \frac{n}{2} \rfloor}\left(\hat Y^\ast_{\pi(4j-3), \rm adj} - \hat Y^\ast_{\pi(4j-2), \rm adj}\right)\left(\hat Y^\ast_{\pi(4j-1), \rm adj} - \hat Y^\ast_{\pi(4j), \rm adj}\right) \\
& \hspace{3em} \times \left(A_{\pi(4j-3)} - A_{\pi(4j-2)}\right)\left(A_{\pi(4j-1)} - A_{\pi(4j)}\right)~, \\
\tilde{\Gamma}_{n, \rm adj} & = \frac{1}{n}\sum_{1 \le i \le 2n: A_i = 1}\hat Y_{i, \rm adj}^\ast - \frac{1}{n}\sum_{1 \le i \le 2n: A_i = 0}\hat Y_{i, \rm adj}^\ast~, \\
\hat Y_{i, \rm adj}^\ast & = Y_i^\ast - \frac{1}{2} (\hat m^\ast_{1, \tilde Y D}(X_i, W_i) + \hat m^\ast_{0, \tilde Y D}(X_i, W_i))~.
\end{align*}
It then follows from Theorem 3.2 in \cite{bai2023covariate} applied on the infeasible outcomes $\hat{Y}^*_{i, \rm adj}$ that this infeasible numerator converges to the desired quantity. It thus remains to show that 
\begin{align}
\label{eq:tau-adj} \hat{\tau}^2_{n, \rm adj} & = \tilde{\tau}^2_{n, \rm adj} + o_P(1) \\
\label{eq:lambda-P-adj} \hat{\lambda}^2_{n, \rm adj} & = \tilde{\lambda}^2_{n, \rm adj} + o_P(1) \\
\label{eq:gamma-P-adj} \hat{\Gamma}_{n, \rm adj} & = \tilde{\Gamma}_{n, \rm adj} + o_P(1)~.
\end{align}
\eqref{eq:tau-adj}--\eqref{eq:gamma-P-adj} can be established exactly as in the proof of Theorem \ref{theorem:varianceestimation}.
\end{proof}

\subsection{Proof of Theorem \ref{thm:zeta}}
\begin{proof}
Define
\begin{align*}
\delta_j^Y & = (A_{\pi(2j - 1)} - A_{\pi(2j)})(Y_{\pi(2j - 1)} - Y_{\pi(2j)}) \\
\delta_j^\zeta & = (A_{\pi(2j - 1)} - A_{\pi(2j)})(\zeta_{\pi(2j - 1)} - \zeta_{\pi(2j)})~.
\end{align*}
It follows from the Frisch-Waugh-Lovell theorem that
\[ \hat \beta_n^Y = \left ( \frac{1}{n} \sum_{1 \leq j \leq n} (\delta_j^\zeta - \hat \Delta_n^\zeta) (\delta_j^\zeta - \hat \Delta_n^\zeta)' \right )^{-1} \frac{1}{n} \sum_{1 \leq j \leq n} (\delta_j^\zeta - \hat \Delta_n^\zeta) \delta_j^Y \]
where
\[ \hat \Delta_n^\zeta = \frac{1}{n} \sum_{1 \leq j \leq n} \delta_j^\psi~. \]
Identical arguments to the proof of Theorem 4.2 in \cite{bai2023covariate} imply
\begin{align*}
    \frac{1}{n} \sum_{1 \leq j \leq n} (\delta_j^\zeta - \hat \Delta_n^\zeta) (\delta_j^\zeta - \hat \Delta_n^\zeta)' & \stackrel{P}{\to} 2E[\var[\zeta_i | X_i]] \\
    \frac{1}{n} \sum_{1 \leq j \leq n} (\delta_j^\zeta - \hat \Delta_n^\zeta) \delta_j^Y & \stackrel{P}{\to} E[\cov[\zeta_i, \tilde Y_i(1) + \tilde Y_i(0)]]~,
\end{align*}
and the limit of $\hat \beta_n^Y$ follows. The result on $\hat \beta_n^D$ follows through the same arguments.

Next, in order to see that $\hat \Delta_n^{\rm adj}$ in Theorem \ref{thm:zeta} is the optimal linear adjustment, note that $\nu_{\rm adj}^2$ only depends on $m_{a, \tilde Y}(X_i, W_i)$ and $m_{a, D}(X_i, W_i)$ through $\nu_{1, \rm adj}^2$. Then for arbitrary linear adjustments $m_{a, \tilde Y}(X_i, W_i) = \zeta_i' \tilde{\beta}^{\tilde Y}(a)$ and $m_{a, D}(X_i, W_i) = \zeta_i' \tilde{\beta}^{D}(a)$ for $a \in \{0, 1\}$, $\nu_{1, \rm adj}^2$ can be re-written as
\begin{align*}
\nu_{1, \rm adj}^2
&= \frac{1}{2} E[ \var[ E\left[ Y_i^\ast(1) + Y_i^\ast(0) | X_i, W_i \right] 
- ( m_{1, \tilde Y D}(X_i, W_i) + m_{0, \tilde Y D}(X_i, W_i) ) | X_i ] ] \\
&= \frac{1}{2} E[ \var[ E\left[ Y_i^\ast(1) + Y_i^\ast(0) | X_i, W_i \right] -  \zeta_i' ( (\tilde \beta^{\tilde Y}(1) + \tilde \beta^{\tilde Y}(0) ) - \Delta(Q) (\tilde \beta^D(1) + \tilde \beta^D(0) )) | X_i ] ] \\
&= \frac{1}{2} E[ E[ (E\left[ Y_i^\ast(1) + Y_i^\ast(0) | X_i, W_i \right] - E\left[ Y_i^\ast(1) + Y_i^\ast(0) | X_i\right] -  \\
& \hspace{3em} (\zeta_i - E[\zeta_i|X_i])' ( (\tilde \beta^{\tilde Y}(1) + \tilde \beta^{\tilde Y}(0) ) - \Delta(Q) (\tilde \beta^D(1) + \tilde \beta^D(0) )) )^2 | X_i ] ] \\
&= \frac{1}{2} E[ (E\left[ Y_i^\ast(1) + Y_i^\ast(0) | X_i, W_i \right] - E\left[ Y_i^\ast(1) + Y_i^\ast(0) | X_i\right] -  \\
& \hspace{3em} (\zeta_i - E[\zeta_i|X_i])' ( (\tilde \beta^{\tilde Y}(1) + \tilde \beta^{\tilde Y}(0) ) - \Delta(Q) (\tilde \beta^D(1) + \tilde \beta^D(0) )) )^2  ]~,
\end{align*}
which minimized when
\begin{align}\label{eq:var_min}
(\beta^{\tilde Y}(1) + \beta^{\tilde Y}(0) ) - \Delta(Q) (\beta^D(1) + \beta^D(0) ) &= (E[\var[\zeta_i|X_i]])^{-1} E[\cov[E[Y_i^\ast(1) + Y_i^\ast(0)|X_i, W_i], \zeta_i | X_i]] \\
&= (E[\var[\zeta_i | X_i]])^{-1} E[ \cov[Y^\ast_i(1) + Y^\ast_i(0), \zeta_i | X_i] ]~,
\end{align}
where the first equality follows from the first order condition of minimizing $\nu_{1, \rm adj}^2$ and the second equality follows by the fact that
\begin{align*}
&E[\cov[E[Y_i^\ast(1) + Y_i^\ast(0)|X_i, W_i], \zeta_i | X_i] - \cov[Y^\ast_i(1) + Y^\ast_i(0), \zeta_i | X_i]] \\
=&E[\cov[E[Y_i^\ast(1) + Y_i^\ast(0)|X_i, W_i] - (Y^\ast_i(1) + Y^\ast_i(0)), \zeta_i | X_i]] \\
=&E[ (E[Y_i^\ast(1) + Y_i^\ast(0)|X_i, W_i] - (Y^\ast_i(1) + Y^\ast_i(0))) \zeta_i ] \\
=& E[E[ (E[Y_i^\ast(1) + Y_i^\ast(0)|X_i, W_i] - (Y^\ast_i(1) + Y^\ast_i(0))) \zeta_i |X_i, W_i ]] \\
=& E[(E[Y_i^\ast(1) + Y_i^\ast(0)|X_i, W_i] - E[Y_i^\ast(1) + Y_i^\ast(0)|X_i, W_i]) \zeta_i  ]\\
=&0~,
\end{align*}
where the fourth equality follows by the fact that $\zeta_i$ depends only on $X_i$ and $W_i$, the rest of the equalities follows by inspection. Finally, note that \eqref{eq:var_min} holds when $\tilde{\beta}^{\tilde Y}(a) = \beta^{\tilde Y}$ and $\tilde{\beta}^{D}(a) = \beta^D$ for $a \in \{0, 1\}$, as desired.
\end{proof}

\subsection{Auxiliary Results}

\begin{lemma}\label{lemma:efficiency-bound}
$\nu^2$ in Theorem \ref{theorem:main} matches the expression of the efficiency bound in Theorem 2 of \cite{frolich2007nonparametric}.
\end{lemma}

\noindent {\sc Proof}: 
The efficiency bound in Theorem 2 of \cite{frolich2007nonparametric} is
\begin{align*}
V &= \frac{1}{P \{C_i = 1\}^2} E\left[ \frac{\var[Y_i|X_i, A_i=1] - 2 \Delta(Q) \cov[Y_i, D_i | X_i, A_i = 1] + \Delta(Q)^2 \var[D_i|X_i, A_i = 1]}{P\{A_i = 1| X_i\}} \right. \\
&\hspace{1.5em} \left. + \frac{\var[Y_i|X_i, A_i=0] - 2 \Delta(Q) \cov[Y_i, D_i | X_i, A_i = 0] + \Delta(Q)^2 \var[D_i|X_i, A_i = 0]}{P\{A_i = 0| X_i\}} \right] \\
& \hspace{1.5em} + \frac{1}{P \{C_i = 1\}^2} E[  ( E[Y_i | X_i, A_i = 1] - E[Y_i | X_i, A_i = 0] \\
& \hspace{5em} - \Delta(Q) E[D_i | X_i, A_i = 1] + \Delta(Q) E[D_i | X_i, A_i = 0] )^2 ] \\
&= \frac{2}{P \{C_i = 1\}^2} E\left[ \var[\tilde Y_i(1)|X_i] - 2 \Delta(Q) \cov[\tilde Y_i(1), D_i(1) | X_i] + \Delta(Q)^2 \var[D_i(1)|X_i] \right. \\
& \hspace{1.5em} \left. + \var[\tilde Y_i(0)|X_i] - 2 \Delta(Q) \cov[\tilde Y_i(0), D_i(0) | X_i] + \Delta(Q)^2 \var[D_i(0)|X_i] \right] \\
& \hspace{1.5em} + \frac{1}{P \{C_i = 1\}^2} E\left[  \left( E[\tilde Y_i(1) | X_i] - E[\tilde Y_i(0) | X_i] - \Delta(Q) E[D_i(1) | X_i] + \Delta(Q) E[D_i(0) | X_i] \right)^2 \right] \\
& = \frac{2}{P \{C_i = 1\}^2} E\left[ 
\var[\tilde Y_i(1) | X_i] + \var[\tilde Y_i(0) | X_i] + \frac{1}{2} E^2 [\tilde Y_i(1) - \tilde Y_i(0) | X_i] \right. \\
& \hspace{1.5em} \left. - 2\Delta(Q) \left( \cov[\tilde Y_i(1), D_i(1) | X_i] +  \cov[\tilde Y_i(0), D_i(0) | X_i] + \frac{1}{2} E[\tilde Y_i(1) - \tilde Y_i(0) | X_i] E[D_i(1) - D_i(0) | X_i ]\right) \right.\\
& \hspace{1.5em} \left. + \Delta(Q)^2 
\left(   \var[D_i(1) | X_i] + \var[D_i(0) | X_i] + \frac{1}{2} E^2 [D_i(1) - D_i(0) | X_i]\right)
\right] \\
&= \frac{2}{P \{C_i = 1\}^2} E\left[ \var[Y_i^\ast(1)|X_i] + \var[Y_i^\ast(0)|X_i] + \frac{1}{2} E^2[Y_i^\ast(1) - Y_i^\ast(0)|X_i] \right]~,
\end{align*}
where the first equality follows by Theorem 2 in \cite{frolich2007nonparametric}, the second equality follows by (\ref{eq:TildeY}) and Assumption \ref{ass:assignment}, the third equality follows by direct calculation, and the fourth equality follows by (\ref{eq:Ystar-a}). Then we have 
\begin{align*}
& \frac{V}{2} - \nu^2 \\
&= \frac{1}{ P \{C_i = 1\}^2} \left( E\left[ \var[Y_i^\ast(1)|X_i] + \var[Y_i^\ast(0)|X_i] + \frac{1}{2} E^2[Y_i^\ast(1) - Y_i^\ast(0)|X_i] \right] \right. \\ 
& \hspace{7.5em} \left. - \var[Y_i^\ast(1)] - \var[Y_i^\ast(0)] + \frac{1}{2} \var[E[Y_i^\ast(1) + Y_i^\ast(0) | X_i]] \right) \\
&= \frac{1}{ 2 P \{C_i = 1\}^2} ( - 2\var[E[Y_i^\ast(1)|X_i]] - 2\var[E[Y_i^\ast(0)|X_i]] + E [E^2[Y_i^\ast(1) - Y_i^\ast(0)|X_i]] \\
& \hspace{7.5em} + \var[E[Y_i^\ast(1) + Y_i^\ast(0) | X_i]]) \\
&= \frac{1}{ 2 P \{C_i = 1\}^2} \left( E [E[Y_i^\ast(1) - Y_i^\ast(0)|X_i]^2] -  \var[E[Y_i^\ast(1) - Y_i^\ast(0) | X_i]] \right) \\
&= \frac{1}{ 2 P \{C_i = 1\}^2} E[Y_i^\ast(1) - Y_i^\ast(0)]^2 \\
&= 0~,
\end{align*}
where the first three equalities follow by inspection, and the fourth equation follows Assumption \ref{ass:Q}(e), and the fact that Assumption \ref{ass:Q}(d) implies $E[Y_i^\ast(1) - Y_i^\ast(0)] = 0$.
\qed

\begin{lemma} \label{lem:fwl}
Suppose $(Y_i, X_{1, i}', X_{2, i}')'$, $1 \leq i \leq n$ is an i.i.d.\ sequence of random vectors, where $Y_i$ takes values in $\mathbf R$, $X_{1, i}, Z_{1, i}$ take values in $\mathbf R^{k_1}$, and $X_{2, i}$ takes values in $\mathbf R^{k_2}$. Consider the linear regression
\[ Y_i = X_{1, i}' \beta_1 + X_{2, i}' \beta_2 + \epsilon_i~. \]
Define $\mathbb X_1 = (X_{1, 1}, \ldots, X_{1, n})'$, $\mathbb Z_1 = (Z_{1, 1}, \ldots, Z_{1, n})'$, $\mathbb X_2 = (X_{2, 1}, \ldots, X_{2, n})'$, $\mathbb X = (\mathbb X_1 \mathbb X_2)$, and $\mathbb Z = (\mathbb Z_1 \mathbb X_2)$. Further define $\mathbb P_2 = \mathbb X_2 (\mathbb X_2' \mathbb X_2)^{-1} \mathbb X_2'$ and $\mathbb M_2 = \mathbb I - \mathbb P_2$. Let $\hat \beta_{1, n}$ and $\hat \beta_{2, n}$ denote the IV estimators of $\beta_1$ and $\beta_2$ using $Z_1$ as an instrument for $X_1$. Define $\hat \epsilon_i = Y_i - X_{1, i}' \hat \beta_{1, n} - X_{2, i}' \hat \beta_{2, n}$. Define
\begin{align*}
\tilde {\mathbb Z}_1 & = \mathbb M_2 \mathbb Z_1 \\
\tilde {\mathbb X}_1 & = \mathbb M_2 \mathbb X_1~.
\end{align*}
Let
\[ \hat \Omega_n = (\mathbb Z' \mathbb X)^{-1} (\mathbb Z' \mathrm{diag}(\hat \epsilon_i^2: 1 \leq i \leq n) \mathbb Z) (\mathbb X' \mathbb Z)^{-1} \]
denote the heteroskedasticity-robust variance estimator of $(\hat \beta_{1, n}, \hat \beta_{2, n})$. Then, the upper-left $k_1 \times k_1$ block of $\hat \Omega_n$ equals
\[ (\tilde {\mathbb Z}_1' \tilde {\mathbb X}_1)^{-1} (\tilde {\mathbb Z}_1' \mathrm{diag}(\hat \epsilon_i^2: 1 \leq i \leq n) \tilde {\mathbb Z}_1) (\tilde {\mathbb X}_1' \tilde {\mathbb Z}_1)^{-1}~. \]
\end{lemma}

\begin{proof}
By the formula for the inverse of a partitioned matrix, the first $k_1$ rows of $(\mathbb Z' \mathbb X)^{-1}$ equal
\[\begin{pmatrix}
(\tilde {\mathbb Z}_1' \tilde {\mathbb X}_1)^{-1} & - (\tilde {\mathbb Z}_1' \tilde {\mathbb X}_1)^{-1} \mathbb Z_1' \mathbb X_2 (\mathbb X_2' \mathbb X_2)^{-1}
\end{pmatrix}~. \]
Furthermore,
\[ \mathbb Z' \mathrm{diag}(\hat \epsilon_i^2: 1 \leq i \leq n) \mathbb Z = \begin{pmatrix}
\mathbb Z_1' \mathrm{diag}(\hat \epsilon_i^2: 1 \leq i \leq n) \mathbb Z_1 & \mathbb Z_1' \mathrm{diag}(\hat \epsilon_i^2: 1 \leq i \leq n) \mathbb X_2 \\
\mathbb X_2' \mathrm{diag}(\hat \epsilon_i^2: 1 \leq i \leq n) \mathbb Z_1 & \mathbb X_2' \mathrm{diag}(\hat \epsilon_i^2: 1 \leq i \leq n) \mathbb X_2
\end{pmatrix}~. \]
The conclusion then follows from elementary calculations.
\end{proof}

\begin{lemma} \label{lem:FE_TSLS}
$$\hat{\alpha}_n^{\rm IV} = \hat \Delta_n^{\rm adj}$$
with working models $\hat m_{a, \tilde Y}(X_i, W_i) = \zeta_i' \hat \beta_n^Y$ and $\hat m_{a, D}(X_i, W_i) = \zeta_i' \hat \beta_n^D$ for $a \in \{0, 1\}$.
\end{lemma}

\begin{proof}
Denote by $\hat \alpha_n^Y$ and $\hat \alpha_n^D$ the ordinary least squares estimators of $\alpha^Y$ and $\alpha^D$ in \eqref{eq:pfe-Y}--\eqref{eq:pfe-D}. The Frisch-Waugh-Lovell theorem implies that 
\begin{eqnarray*}
\hat \alpha_n^Y &=&  \hat \psi_n^{\rm adj}(1) - \hat \psi_n^{\rm adj}(0) \\
\hat \alpha_n^D &=& \hat \phi_n^{\rm adj}(1) - \hat \phi_n^{\rm adj}(0)~.
\end{eqnarray*}
It thus follows immediately that
\[ \frac{\hat \alpha_n^Y}{\hat \alpha_n^D} = \hat \Delta_n^{\rm adj} \] when $\hat m_{a, \tilde Y}(X_i, W_i) = \zeta_i' \hat \beta_n^Y$ and $\hat m_{a, D}(X_i, W_i) = \zeta_i' \hat \beta_n^D$ for $a \in \{0, 1\}$. Next, the subvector formula for IV implies that
\[ \hat \alpha^{\rm IV}_n = \left ( \frac{1}{n} \sum_{1 \leq i \leq 2n} \tilde A_i D_i \right )^{-1} \left ( \frac{1}{n} \sum_{1 \leq i \leq 2n} \tilde A_i Y_i \right )~, \]
where $\tilde A_i$ is the residual in the projection of $A$ on $\zeta$ and fixed effects. 
From this it follows that 
\[ 
\hat{\alpha}^{\rm IV}_n = \frac{\left( \frac{1}{n} \sum_{1 \leq i \leq 2n} \tilde{A}_i^2 \right)^{-1} \left( \frac{1}{n} \sum_{1 \leq i \leq 2n} \tilde{A}_i Y_i \right)}{\left( \frac{1}{n} \sum_{1 \leq i \leq 2n} \tilde{A}_i^2 \right)^{-1} \left( \frac{1}{n} \sum_{1 \leq i \leq 2n} \tilde{A}_i D_i \right)} = \frac{\hat{\alpha}_n^Y}{\hat{\alpha}_n^D} = \hat{\Delta}_n^{\rm adj}~,
\]
as desired.
\end{proof}

\end{document}